\documentclass[10pt, a4paper, onecolumn, conference, compsocconf]{IEEEtran}
\usepackage{graphicx}
\usepackage{url}
\usepackage{times}
\usepackage[cmex10]{amsmath}
\usepackage[ruled,vlined]{algorithm2e}

\newenvironment{list1}{\begin{list}{$\bullet$}
{\topsep 0 pt \parsep 0 pt \partopsep 0 pt \itemsep 0
pt}}{\end{list}}

\newtheorem{defi}{Definition}
\newtheorem{theo}{Theorem}
\newtheorem{prop}{Proposition}

\begin{document}
\title{A Game-theoretic Approach for Synthesizing Fault-Tolerant \\Embedded Systems}


\author{
\IEEEauthorblockN{Chih-Hong Cheng\authorrefmark{1},
                    Harald Ruess\authorrefmark{2},
                    Alois Knoll\authorrefmark{1}
                    Christian Buckl\authorrefmark{2}}
\IEEEauthorblockA{\authorrefmark{1}
Department of Informatics, Technische Universit\"{a}t M\"{u}nchen, Germany}
\IEEEauthorblockA{\authorrefmark{2}fortiss GmbH, Germany
\\Email:\{chengch,knoll\}@in.tum.de, \{ruess,buckl\}@fortiss.org }}


%Telephone: (800) 555--1212, Fax: (888) 555--1212}


\maketitle

\begin{abstract}
In this paper, we present an approach for fault-tolerant synthesis by combining predefined patterns for fault-tolerance with algorithmic game solving. A non-fault-tolerant system, together with the relevant
fault hypothesis and fault-tolerant mechanism templates in a pool
are translated into a distributed game, and we perform an incomplete search of strategies to cope with undecidability.
The result of the game is translated back to executable code concretizing fault-tolerant mechanisms
using constraint solving.
The overall approach is implemented to a prototype tool chain and is illustrated using examples.
\end{abstract}

%\begin{IEEEkeywords}
%component; formatting; style; styling;

%\end{IEEEkeywords}

% \fi
\IEEEpeerreviewmaketitle
%-------------------------------------------------------------------------

\section{Introduction}

%An automatic method, which enables to equip suitable fault-tolerant mechanisms on a
%non-fault-tolerant system with least human-intervention, such that the resulting annotated
%system is able to resist the faults defined by the fault model, is always desired.
In this paper, we investigate methods to perform automatic fault-tolerant (FT
for short) synthesis under the context of embedded systems, where our goal is to generate executable
code which can be deployed on dedicated hardware platforms.

Creating such a tool supporting the fully-automated process is very challenging as the inherent complexity is high: %of the problem:
bringing FT synthesis from theory to practice means solving a
problem consisting of (a) interleaving semantics, (b) timing, (c) fault-tolerance, (d) dedicated
features of concrete hardware, and optionally, (e) the code generation framework.
To generate tamable results, we first
constrain our problem space to some simple
yet reasonable scenarios (sec.~\ref{sec.motivating.scenarios}).
Based on these scenarios we can start system modeling (sec.~\ref{sec.system.modeling}) taking into account all above mentioned aspects.
%It's also important for the initial model under processing to be close to models
%commonly used in the embedded or real-time system domain, such that engineers are comfortable with them.

%We find it important to
To proceed further, we find it important to observe the approach nowadays to understand the need: for engineers working on
ensuring fault-tolerance of a system,
once the corresponding fault model is decided, a common approach is to select some
fault-tolerant patterns~\cite{hanmer:2007:patterns} (e.g., fragments of executable code) from a pattern pool. Then
engineers must fine-tune these mechanisms, or fill in unspecified information in the
patterns to make them work as expected. With the above scenario in mind, apart from generating complete FT mechanisms from specification,
our synthesis technique emphasizes automatic selection of predefined FT patterns and automatic tuning such that
details (e.g., timing) can be filled without human intervention.
This also reduces a potential problem where unwanted FT mechanisms are synthesized due to under-specification.
Following the statement, we translate the system model, the fault hypothesis, and the set of available FT patterns
into a distributed game~\cite{mohalik:2003:distributed} %using safe abstractions
(sec.~\ref{sec.frontend.translation}), and a
strategy generated by the game solver can be interpreted as a selection of FT patterns together with guidelines of tuning.
%, and precise timing information are annotated during platform analysis.

For games, it is known that solving distributed games is, in most cases, undecidable~\cite{mohalik:2003:distributed}.
To cope with undecidability, we restrict ourselves to the effort of finding positional strategies (mainly for reachability games).
We argue that finding positional strategies is still practical, as the selected FT patterns may introduce
additional memory during game creation.
Hence, a positional strategy (by pattern selection) combined with selected FT patterns generates mechanisms using memory.
By posing this restriction, the problem of finding a strategy of the game (for control) is NP-Complete (sec.~\ref{sec.game.solving}),
and searching techniques (e.g., SAT translation or combining forward search with BDD) are thus applied to assist the finding of solutions.

%(or equals to a finite-state memory strategy)
%We mention two algorithms to solve the simplified problem:
%the first is a bounded forward search combined with BDDs, and the second is a translation algorithm to SAT.

%A strategy generated by the game solver can be interpreted as a solution instance of
%the FT synthesis problem (however, it may not be fulfilled in the implementation).
The final step of the automated process is to translate the result of synthesis
back to concrete implementation:
the main focus is to ensure that
the newly synthesized mechanisms do not change the implementability  of the original system (i.e., the new system is schedulable).
Based on our modeling framework, this problem can be translated to a linear constraint system, which can be solved efficiently
by existing tools.
%As the underlying hardware is specified beforehand, it is undesirable (yet always possible) to perform a full system rescheduling.
%To avoid full-rescheduling, we propose an algorithmic criterion based on constraint solving to decide
%whether it is possible to perform a partial adjustment.

%\footnote{Due to lack of space,
%this part is only listed in Appendix B and will not be included.}.

To evaluate our methods, we have created our prototype software,
which utilizes the model-based approach to facilitate the design, synthesis, and code generation for fault-tolerant embedded systems.
We demonstrate two small yet representative examples with our tool for a proof-of-concept (sec.~\ref{sec.case.study});
these examples indicate the applicability of the approach.
Lastly, we conclude this paper with an overview of related work (sec.~\ref{sec.related.work}) and a brief summary including the flow of our approach (sec.~\ref{sec.concluding.remarks}).

%To our knowledge, this is the first report where automatic FT synthesis for concrete embedded systems is possible.
%: with \textsc{Gecko}
%we can perform automatic synthesis, and generate executable code for simple systems running on dedicated but commonly seen
%hardware and buses in embedded systems.

%1. Synthesis is in general undecidable quote. If we restrict ourselves to positional strategies then NP-C
%2. State space may become large, as (1) interleaving (2) fault-tolerance
%Our goal: executable code on embedded systems (don't say anything others)

%Draw a picture: overall methodology, but focus on the first stage.
%But the overall methodology is implemented.

%What is the benefit

%State-of-the art.

%Definition: Give Fig.1 as PISEM.
%PISEM is derived from the real-time
%Put the precise definitions in the appendix.

\section{Motivating Scenario\label{sec.motivating.scenarios}}

\subsection{Adding FT Mechanisms to Resist Message Loss\label{subsec.example.adding.ft}}

We give a motivating scenario in embedded systems to facilitate our mathematical definitions.
The simple system described in Figure~\ref{fig:Motivating.Example} contains two \emph{processes} $\mathcal{A}$, $\mathcal{B}$
and one bidirectional \emph{network} $\mathcal{N}$.
Processes $\mathcal{A}$ and $\mathcal{B}$ start executing sequential actions together with a looping period of $100 ms$. %\footnote{When looping periods of processes are not identical, we treat all running processes as they were running on the same period $T$, where $T$ is the least common multiple of all periods.}.
In each period, $\mathcal{A}$ first reads an input using a sensor to variable $m$, followed by sending the result to the network $\mathcal{N}$ using the action $\verb"MsgSend(m)"$, and outputing the value (e.g., to a log).

\begin{figure}[t]
\centering
 \includegraphics[width=0.7\columnwidth]{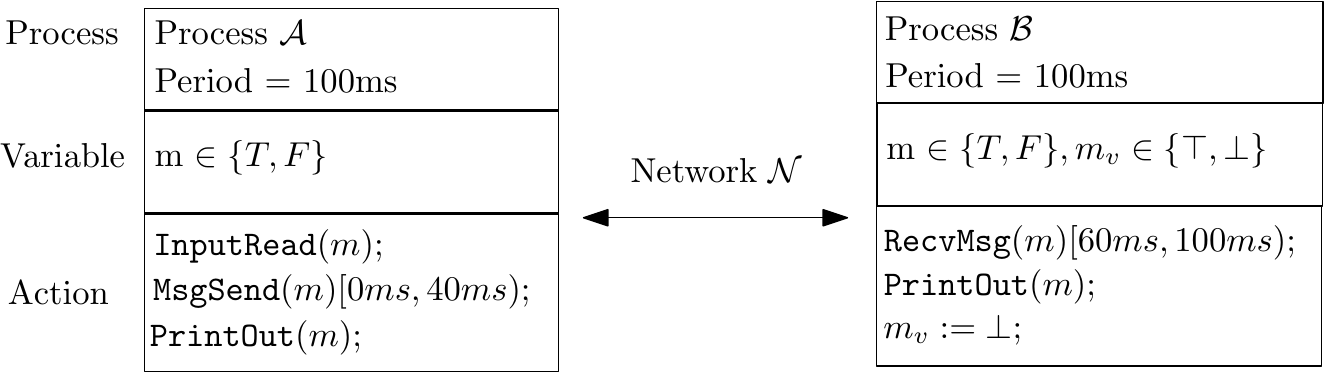}
  \caption{An example for two processes communicating over an unreliable network.}
 \label{fig:Motivating.Example}
\end{figure}

% Restate the example for control by wire control-break.

In process $\mathcal{A}$, for the action $\verb"MsgSend(m)"$, a message containing value of $m$ is forwarded to $\mathcal{N}$, and $\mathcal{N}$ broadcasts the value to all other processes which contain a variable named $m$, and set the variable $m_v$ in $\mathcal{B}$ as $\top$ (indicating that the content is valid).
However, $\mathcal{A}$ is unaware whether the message has been sent successfully: the network component $\mathcal{N}$ is unreliable, which has a faulty behavior of \emph{message loss}.  The fault type and the frequency of the faulty behavior are specified in the \emph{fault model}:
in this example for every complete period ($100 ms$), at most one message loss can occur.

In $\mathcal{B}$, its first action $\verb"RecvMsg(m)"$ has a property describing an interval $[60, 100)$, which specifies the \emph{release time} and \emph{deadline} of this action to be $60 ms$ and $100 ms$, respectively. By posing the release time and the deadline, in this example, $\mathcal{B}$ can finalize its decision whether it has received the message $m$ successfully using the equality constraint $(m_v = \bot)$, provided that the time interval $[40, 60)$ between (a) deadline of $\verb"MsgSend(m)"$ and (b) release time of $\verb"RecvMsg(m)"$ overestimates the \emph{worst case transmission time} for a message to travel from $\mathcal{A}$ to $\mathcal{B}$. After $\verb"RecvMsg(m)"$, it outputs the received value (e.g., to an actuator).

Due to the unreliable network, it is easy to observe that two output values may not be the same.
Thus the \emph{fault-tolerant synthesis} problem in this example is to
perform suitable modification on $\mathcal{A}$ and $\mathcal{B}$, such that two output values from $\mathcal{A}$ and $\mathcal{B}$
are the same at the end of the period, regardless of the disturbance from the network.

\subsection{Solving Fault-Tolerant Synthesis by Instrumenting Primitives}

To perform FT synthesis in the example above, our method is to introduce several slots (the size of slots are fixed by the designer)
between actions originally specified
in the system. For each slot, an atomic operation can be instrumented, and these actions are
among the pool of predefined \emph{fault-tolerant primitives},
consisting of message sending, message receiving, local variable modifications, or \verb"null-op"s.
Under this setting we have created a game, as the original transitions in the fault-intolerant system combined with all FT
primitives available constitute the controller (player-0) moves, and the triggering of faults and the networking can be modeled as environment (player-1) moves.

%Note that when many choices are available, it may run the risk of generating undesired results.
%How to examine the original system behavior is not disturbed?

\section{System Modeling\label{sec.system.modeling}}

\subsection{Platform Independent System Execution Model\label{sub.sec.PISEM}}

We first define the execution model where timing information is included;
it is used for specifying embedded systems and is
linked to our code-generation framework. In the definition, for ease of understanding we also give each term intuitive explanations.

\begin{defi}
Define the syntax of the \textbf{Platform-Independent System Execution Model (PISEM)} be $\mathcal{S} = (\mathcal{A}, \mathcal{N}, \mathcal{T})$.
\begin{list1}
\item $\mathcal{T} \in \mathbf{Q}$ is the replication period of the system.
\item $\mathcal{A} = \bigcup_{i = 1\ldots n_A} \mathcal{A}_i$ is the set of processes, where in $\mathcal{A}_{i} = (V_{i} \cup V_{env_i}, \overline{\sigma_i})$,
    \begin{itemize}
        \item $V_i$ is the set of variables, and $V_{env_i}$ is the set of environment variables. For simplicity assume that $V_i$ and $V_{env_i}$ are of integer domain.
        \item  $\overline{\sigma_i} := \sigma_1 [\alpha_1, \beta_1); \ldots; \sigma_j [\alpha_j, \beta_j); \ldots ; \sigma_{k_i} [\alpha_{k_i}, \beta_{k_i})$ is a sequence of actions.
           \begin{itemize}
                \item $\sigma_j := \verb"send"(pre, index, n, s, d, v, c)\; | \; a \leftarrow \verb"e" \; |\; \verb"receive"(pre, c)$ is an atomic action (action pattern), where
                     \begin{itemize}
                         \item $a, c\in V_i$,
                         \item $\verb"e"$ is function from $V_{env_x} \cup V_i$ to $V_{i}$ (this includes \verb"null-op"),
                         \item $pre$ is a conjunction of over equalities/inequalities of variables,
                         \item $s, d \in \{1, \ldots, n_A\}$ represents the source and destination,
                         \item $v \in V_d$ is the variable which is expected to be updated in process $d$,
                         \item $n \in \{1, \ldots, n_N\}$ is the network used for sending, and
                        \item $index \in \{1,\ldots, size_n\}$ is the index of the message used in the network.
                    \end{itemize}
                 \item $[\alpha_j, \beta_j)$ is the execution interval, where
                          $\alpha_j \in \mathbf{Q}$ is the release time and
                          $\beta_j \in \mathbf{Q}$ is the deadline.
            \end{itemize}
    \end{itemize}
\item $\mathcal{N} = \bigcup_{i = 1\ldots n_N} \mathcal{N}_i$, $\mathcal{N}_i = (\mathcal{T}_i, size_i)$ is the set of network.
    \begin{itemize}
        \item $\mathcal{T}_i: \mathbf{N} \rightarrow \mathbf{Q}$ is a function which maps the index (or priority) of a message to the worst case message transmission time (WCMTT).
        \item $size_i$ is the number of messages used in $\mathcal{N}_i$.
        %\item \begin{small}$tran_i:  \{\verb"true", \verb"false"\} \times \{1, \ldots, n_A\}^{2}\times \bigcup_{i = 1,\ldots, n_A} (V_{i} \cup V_{env_i}) \times \mathbf{Z} \times \mathbf{Q} \times \{1, \ldots ,size_i\} \rightarrow \{\verb"true", \verb"false"\} \times \{1, \ldots, n_A\}^{2}\times \bigcup_{i = 1,\ldots, n_A} (V_{i} \cup V_{env_i}) \times \mathbf{Z} \times \mathbf{Q} \times \{1, \ldots ,size_i\}$\end{small} is the network transition relation for message forwarding.
    \end{itemize}
\end{list1}
\end{defi}

%We give the intuitive meaning of PISEM:
% Note that here for the ease of formulation, we use a broadcasting network; this constraint is

\noindent \textbf{[Example]} Based on the above definitions, the system under
execution in section~\ref{subsec.example.adding.ft} can be easily modeled by PISEM:
let $\mathcal{A}$, $\mathcal{B}$, and $\mathcal{N}$ in section~\ref{subsec.example.adding.ft} be renamed in a PISEM as $\mathcal{A}_1$, $\mathcal{A}_2$,
and $\mathcal{N}_1$. For simplicity, we use $\mathcal{A}.j$ to represent the variable $j$ in process $\mathcal{A}$, assume that the network transmission time is $0$, and let $v_{env}$ contain only one variable $v$ in $\mathcal{A}_1$.
Then in the modeled PISEM, we have $\mathcal{N}_1 = (f:\mathbf{N}\rightarrow0, 1)$, $\mathcal{T} = 100$, and the action sequence of process $\mathcal{A}_1$ is
\begin{small}
\[
m \leftarrow \verb"InputRead"(v) [0, 40); \verb"send"(true, 1, 1, 1, 2, m, \mathcal{A}_1.m)[0, 40); v \leftarrow \verb"PrintOut"(m) [40, 100);
\]
\end{small}
For convenience, we use
 $|\overline{\sigma_{i}}|$ to represent the length of the action sequence $\overline{\sigma_{i}}$, $\sigma_j.deadline$ to represent
 the deadline of $\sigma_j$, and
 $iSet(\overline{\sigma_i})$ to represent a set containing (a) the set of subscript numbers in $\overline{\sigma_i}$ and (b) $|\overline{\sigma_i}| +1$, i.e., $\{1,\ldots,k_i, k_i +1\}$.

\begin{defi}
The configuration of $\mathcal{S}$ is $(\bigwedge_{i = 1 \ldots n_A} (v_i, v_{env_i}, \Delta_{next_i}),
\bigwedge_{j = 1 \ldots n_N} (occu_j, s_j, d_j, var_j, c_j, t_j, ind_j), t)$, where
    \begin{itemize}
        \item $v_i$ is the set of the current values for the variable set $V_i$,
        \item $v_{env_i}$ is the set of the current values for the variable set $V_{env_i}$,
        \item $\Delta_{next_i} \in [1, |\overline{\sigma_i}|+1 ]$ is the next atomic action index taken in
                $\overline{\sigma_i}$\footnote{Here an interval $[1, |\overline{\sigma_i}|+1 ]$ is used for the introduction of FT mechanisms described later.},
        \item $occu_j \in \{\verb"false",\verb"true"\}$ is for indicating whether the network is busy,
        \item $s_j, d_j \in \{1,\ldots, n_A\}$,
        \item $var_j \in \bigcup_{i = 1,\ldots, n_A} (V_{i} \cup V_{env_i})$,
        \item $c_j \in\mathbf{Z}$ is the content of the message,
        \item $ind_j \in \{1,\ldots, size_j\}$ is the index of the message occupied in the network,
        \item $t_j$ is the reading of the clock used to estimate the time required for transmission,
        \item $t$ is the current reading of the global clock.
    \end{itemize}
\end{defi}

% For $v$, we use the notation $v'$ for the updated value,

The change of configuration is caused by the following operations.
\begin{enumerate}
    \item (\emph{Execute local action}) For machine $i$, let $s$ and $j$ be the current configuration for $var$ and $\Delta_{next_i}$, and $v_i$, $v_{env_i}$ are current values of $V_i$ and $V_{env_i}$. If $j = |\overline{\sigma_{i}}|+1$ then do nothing (all actions in $\overline{\sigma_i}$ have been executed in this cycle); else the action $\sigma_j := var \leftarrow \verb"e" [\alpha_j, \beta_j)$ updates $var$ from $s$ to $\verb"e"(v_i, v_{env_i})$, and changes $\Delta_{next_i}$ to $min\{x| x \in iSet(\overline{\sigma_i}), x > j \}$. This action should be executed between the time interval $t \in [\alpha_j, \beta_j)$.

    \item (\emph{Send to network}) For machine $i$, let $s$ and $j$ be the current configuration for $var$ and $\Delta_{next_i}$. If $j = |\overline{\sigma_{i}}|+1$ then do nothing; else  the action $\sigma_j := \verb"send"(pre, index, n, s, d, v, c)[\alpha_j, \beta_j)$ should be processed between the time interval $t \in [\alpha_j, \beta_j)$, and changes $\Delta_{next_i}$ to $min\{x| x \in iSet(\overline{\sigma_i}), x > j \}$.
        \begin{list1}
            \item When $pre$ is evaluated to true (it can be viewed as an \verb"if" statement), it then checks the condition $occu_n = false$: if the condition holds, it updates network $n$ with value $(occu_n, s_n, d_n, var_n, c_n, t_n, ind_n):=(true, i, d, v, c, 0, index)$. Otherwise it blocks until the condition holds.
            \item When $pre$ is evaluated to false, it skips the sending.
        \end{list1}

    \item (\emph{Process message}) For network $j$, for configuration $(occu_j, s_j, d_j, var, c_j, t_j, ind_j)$ if $occu_j = true$, then during $t_j < \mathcal{T}_j(ind_j)$, a transmission occurs, which updates $occu_j$ to \verb"false", $A_{d_j}.var$ to $c_j$, and $A_{d_j}.var_{v}$  to \verb"true".

    \item (\emph{Receive}) For machine $i$, let $s$ and $j$ be the current configuration for $c$ and $\Delta_{next_i}$. If $j = |\overline{\sigma_{i}}|+1$ then do nothing; else for $\verb"receive"(pre, c)[\alpha_j, \beta_j)$ in machine $i$, it is processed between the time interval $t \in [\alpha_j, \beta_j)$ and changes $\Delta_{next_i}$ to $min\{x| x \in iSet(\overline{\sigma_i}), x > j \}$\footnote{In our formulation,
        the \texttt{receive}$(pre, c)$ action can be viewed as a syntactic sugar of \texttt{null-op};
        its purpose is to facilitate the matching of send-receive pair with variable $c$.}.
        %when $pre$ is evaluated to true,
        % It does not matter on receive, it is just used as an timer!

    \item (\emph{Repeat Cycle}) %All clocks are executed uniformly, and
    When $t = \mathcal{T}$, $t$ is reset to $0$, and for all $x\in\{1, \ldots, n_A\}$, $\Delta_{next_x}$ are reset to $1$.

    %\item For machine $i$, for all $j= 1\ldots n$, let $s[j]$ be the current configuration for $a[j]$.
    % Then an action $\sigma_j := \verb"receive"(a[i \oplus 1],\ldots, a[i \oplus (n - 1)])$ performs the following updates.
\end{enumerate}

Notice that by using this model to represent the embedded system under analysis, we make the following assumptions:
\begin{itemize}
\item \emph{All processes and networks in $\mathcal{S}$ share a globally synchronized clock}.
        Note that this assumption can be fulfilled in many hardware platforms, e.g., components
        implementing the IEEE 1588~\cite{eidson2002ieee} protocol.
\item  For all actions $\sigma$, $\sigma.deadline < \mathcal{T}$; for all send actions $\sigma := \verb"send"(pre, index, n, s, d, v, c)$, $\sigma.deadline + \mathcal{T}_n(index)< \mathcal{T}$, i.e., all processes and networks should finish its work within one complete cycle.
\end{itemize}

\subsection{Interleaving Model (IM)}

Next, we establish the idea of interleaving model (IM) which is used to offer an intermediate representation to
bridge PISEM and game solving, such that (a) it captures the execution semantics of PISEM without explicit statements of timing, and
(b) by using this model it is easier to connect to the standard representation of games.

%\subsubsection{Translation from PISEM to IM}
\begin{defi}
Define the syntax of the \textbf{Interleaving Model (IM)} be $S_{IM} = (A, N)$.
\begin{list1}
\item $A = \bigcup_{i = 1\ldots n_A} A_i$ is the set of processes, where in $A_{i} = (V_{i} \cup V_{env_i}, \overline{\sigma_i})$,
    \begin{itemize}
        \item $V_i$ is the set of variables, and $V_{env_i}$ is the set of environment variables.
        \item  $\overline{\sigma_i} := \sigma_1 [\wedge_{m = 1 \ldots n_A} [pc_{1,m_{low}}, pc_{1,m_{up}})]; \ldots; \sigma_j [\wedge_{m = 1 \ldots n_A} [pc_{j, m_{low}}, pc_{j, m_{up}})]; \ldots ; \sigma_{k_i}[\wedge_{m = 1 \ldots n_A} [pc_{k_i, m_{low}}, pc_{k_i, m_{up}})]$ is a fixed sequence of actions.
            \begin{itemize}
                \item $\sigma_j := \verb"send"(pre, index, n, s, d, v, c)\; | \;\verb"receive"(pre, c)\; | \;a \leftarrow \verb"e"  $ is an atomic action, where $a, c, e, pre, v, n,\\ s, d$ are defined similarly as in PISEM.

                \item For $\sigma_j$, $\forall m \in \{1, \ldots, n_A\}$, $pc_{j, m_{low}},pc_{j, m_{up}} \in \{1, \ldots, |\overline{\sigma_m}|+2 \}$ is the lower and the upper bound (PC-precondition interval) concerning
                \begin{enumerate}
                    \item  precondition of program counter in machine $k$, when $m\neq i$.
                    \item  precondition of program counter for itself, when $m = i$.
                \end{enumerate}

            \end{itemize}
    \end{itemize}
\item $N = \bigcup_{i = 1\ldots n_N} N_i$, $N_i = (T_i, size_i)$ is the set of network.
    \begin{itemize}
        \item $T_i: \mathbf{N} \rightarrow \bigwedge_{m = 1 \ldots n_A} (\{1, \ldots, |\overline{\sigma_m}|+2 \}, \{1, \ldots, |\overline{\sigma_m}|+2 \})$ is a function which maps the index (or priority) of a message to the PC-precondition interval of other processes.
        \item $size_i$ is the number of messages used in $\mathcal{N}_i$.
    \end{itemize}
\end{list1}
\end{defi}

\begin{defi}
The configuration of $S_{IM}$ is $(\bigwedge_i (v_i, v_{env_i}, \Delta_{next_i}), \bigwedge_j (occu_j, s_j, d_j, c_j))$, where
$v_i, v_{env_i}, \Delta_{next_i}, occu_j, \\s_j, d_j, c_j$ are defined similarly as in PISEM.
\end{defi}

The change of configurations in IM can be interpreted analogously to PISEM; we omit details here but mention three differences:
\begin{enumerate}
    \item For an action $\sigma_j$ having the precondition $[\wedge_{m = 1 \ldots n_A} [pc_{j, m_{low}}, pc_{j, m_{up}})]$, it should be executed between $pc_{j, m_{low}} \leq \Delta_{next_m} < pc_{j, m_{up}}$, for all $m$.
        %Note that in the definition $pc_{j, m_{low}}, pc_{j, m_{up}}$ are within $\{1, \ldots, |\overline{\sigma_j}|+2 \}$, not $\{1, \ldots, |\overline{\sigma_j}|+1 \}$,
        %for the reason that other actions may execute during the period between (a)
        %$|\overline{\sigma_j}|$ finishes its last action and (b) the end of the period: using the domain $\{1, \ldots, |\overline{\sigma_j}|+1 \}$ is insufficient to capture this phenomenon.
    \item For processing a message, constraints concerning the timing of transmission
            in PISEM are replaced by referencing the PC-precondition interval of other processes in IM, similar to $1$.
    \item The system repeats the cycle when  $\forall x \in \{1, \ldots, n_A\}$, $\Delta_{next_x} = |\overline{\sigma_x}| + 1$ and $\forall x \in \{1, \ldots, n_N\}$, $occu_x = \verb"false"$.
\end{enumerate}

\section{Games}

For the proof of complexity results, we use similar notations in~\cite{mohalik:2003:distributed} to define a distributed game.
Intuitively, distributed games are games formulating multiple processes with no interactions among themselves but only with the environment.

\subsection*{(Local) Games\label{subsec.local.games}}
\newcommand{\Occ}{\ensuremath{\textrm{Occ}}}
\newcommand{\Inf}{\ensuremath{\textrm{Inf}}}
\newcommand{\attr}{\ensuremath{\textrm{attr}}}
\newcommand{\Attr}{\ensuremath{\textrm{Attr}}}
\newcommand{\N}{\mathbf{N}}
A {\em game graph} or {\em arena} is a directed graph $G=(V_0 \uplus V_1,E)$
whose nodes are partitioned into two classes $V_0$ and $V_1$. We only consider the case of two players in
the following and call them player $0$ and player $1$ for simplicity.
A {\em play} starting from node $v_0$ is simply a maximal path $\pi=v_0 v_1 \ldots$ in $G$
where we assume that player $i$ determines the {\em move} $(v_k,v_{k+1})\in E$ if $v_k\in V_i$ ($i\in\{0,1\}$).
With $\Occ(\pi)$ we denote the set of nodes visited by a play $\pi$.
A {\em winning condition} defines when a given play $\pi$ is {\em won} by player $0$;
if $\pi$ is not won by player $0$, it is won by player $1$.
A node $v$ is won by player $i$ if player $i$ can always
choose his moves in such a way that he wins any resulting play starting from $v$.
%;the sets of nodes won by player $i$ are denoted by  $W_i$ ($i\in\{0,1\}$).

\subsection*{Distributed Games\label{subsec.distributed.games}}
We use %similar
notations by Mohalik and Walukiewicz~\cite{mohalik:2003:distributed} to define a distributed game.
From now on we call the a game graph defined in sec.~\ref{subsec.local.games} a \emph{local game graph}.
\begin{defi}
For all $i \in \{1,\ldots, n\}$, let
$G_i = (V_{0_i} \uplus V_{1_i}, E_i )$ be a local game graph with the restriction that it is bipartite. Define a \emph{distributed game} to be
$\mathcal{G} = (\mathcal{V}_0 \uplus \mathcal{V}_1, \mathcal{E}, Acc \subseteq (\mathcal{V}_{0} \uplus \mathcal{V}_{1})^{\omega})$:
\begin{itemize}
    \item $\mathcal{V}_1 = V_{1_1} \times \ldots \times V_{1_n}$ is the set of player $1$ (environment) vertices.
    \item $\mathcal{V}_0 = (V_{0_1} \uplus V_{1_1}) \times \ldots \times (V_{0_n} \uplus V_{1_n}) \setminus \mathcal{V}_1$ is the set of player $0$ (control) vertices.
            \begin{itemize}
                \item For a vertex $x=(x_1, \ldots, x_n)$, we use the function $proj(x, i)$  to retrieve the $i$-th component $x_i$, and use $proj(X, i)$ to retrieve the $i$-th component for a set of vertices $X$.
            \end{itemize}
    \item Let $(x_1, \ldots, x_n), (x'_1, \ldots, x'_n) \in \mathcal{V}_0 \uplus \mathcal{V}_1$, then define $\mathcal{E}$ as follows:
    \begin{itemize}
        % Use the definition of Walukiwitz directly (not very lovely definition)
       % \item If $(x_1, \ldots, x_n) \in \mathcal{V}_0$,  $((x_1, \ldots, x_n), (x'_1, \ldots, x'_n)) \in \mathcal{E}$ if and only if \\$\exists ! i. (((x_i, x_i') \in E_i \wedge x_i \in V_{0_i})\wedge(\forall j \neq i. \;x_j = x_j'))$.

        %    \begin{list1}
        %        \item For convenience, we use the function $local: \mathcal{E} \rightarrow \{0,\ldots, n\}$: given \\$((x_1, \ldots, x_n), (x'_1, \ldots, x'_n))\in \mathcal{E}$, it returns $0$ if it is a player-$1$ transition, and returns the index $i$ which $(x_i, x_i') \in E_i$ if it is a player-$0$ transition.
        %    \end{list1}
        %\item If $(x_1, \ldots, x_n) \in \mathcal{V}_1$,   $((x_1, \ldots, x_n), (x'_1, \ldots, x'_n)) \in \mathcal{E}$ if\\ $\forall i\in \{1, \ldots, n\} .\;(x_i, x'_i) \in E_i$.

        \item If $(x_1, \ldots, x_n) \in \mathcal{V}_0$,  $((x_1, \ldots, x_n), (x'_1, \ldots, x'_n)) \in \mathcal{E}$ if and only if \\$\forall i. ( x_i \in V_{0_i} \rightarrow (x_i, x_i') \in E_i)\wedge\forall j. \;
            (x_j \in V_{1_j} \rightarrow x_j = x_j')$.
            %\begin{list1}
                %\item For convenience, we use the function $local: \mathcal{E} \rightarrow \{0,\ldots, n\}$: given \\$((x_1, \ldots, x_n), (x'_1, \ldots, x'_n))\in \mathcal{E}$, it returns $0$ if it is a player-$1$ transition, and returns the index $i$ which $(x_i, x_i') \in E_i$ if it is a player-$0$ transition.
            %\end{list1}
         \item For $(x_1, \ldots, x_n) \in \mathcal{V}_1$, if $((x_1, \ldots, x_n), (x'_1, \ldots, x'_n)) \in \mathcal{E}$,
         then for every $x_i$, either $x_i = x'_i$ or $x'_i \in V_{0_i}$, and moreover $(x_1, \ldots, x_n) \neq (x'_1, \ldots, x'_n)$.
        %\begin{enumerate}
            % \item $\exists i \in \{1 \ldots n\}.  (x_i, x'_i) \in E_i$ and
            % \item $\forall i\in \{1 \ldots n\} .\;(x_i = x'_i \vee (x_i, x'_i) \in E_i)$.
            % \item $\sigma.deadline + \mathcal{T}_n(index)< \mathcal{T}$
              %\item $\exists i\in \{1, \ldots, n\} .\; x_i \neq x'_i$.
        %\end{enumerate}
        % This definition ensures that for an environment transition, at least one local transition should be executed.
    \end{itemize}
    \item $Acc$ is the acceptance condition.
\end{itemize}
\end{defi}

%The formulation here ensures that (a) for a player $1$ transition, every
%local game graph should proceed with a move, and (b) for a player $0$ transition, only one local game can move.
%Compared to the definition in~\cite{mohalik:2003:distributed} where (a)
%at least one local transition should be executed in the environment move and (b) multiple local moves can constitute a control move, two
%formulations can be translated bidirectionally. % in equivalent forms. (equivalence of games; need more time to explain, so better skip)
%Lastly, in both formulations the environment may have fewer choices to move, compared to the free product of local games.
%, the formulation in~\cite{mohalik:2003:distributed} can be translated
%by introducing edges as self-loops on player-1 vertices in each local game.

In a distributed game $\mathcal{G} = (\mathcal{V}_0 \uplus \mathcal{V}_1, \mathcal{E}, Acc)$, a play is defined analogously as defined in local games: a \emph{play} starting from node $v_0$ is a maximal path $\pi=v_0 v_1 \ldots$ in $\mathcal{G}$ where player $i$ determines the \emph{move} $(v_k,v_{k+1})\in \mathcal{E}$ if $v_k\in \mathcal{V}_i$ ($i\in\{0,1\}$).
%For convenience, denote the set of all plays starting from vertex $v_0$ be $\pi(v_0)$.
%Also we consider \textbf{fair executions} in a game, i.e., a local process does not always starve and eventually executes its action when it is possible.

A \emph{distributed strategy} of a distributed game for player $0$ is a tuple of functions $\xi = \langle f_1, \ldots, f_n\rangle$, where each function $f_i : (V_{0_i} \uplus V_{1_i})^{*}\times V_{0_i} \rightarrow (V_{0_i} \uplus V_{1_i})$ is a local strategy which decides the updated location of the local game $i$ based on (a) its observable history of local game $i$ and (b) current position of local game $i$. Lastly, we call a distributed strategy \emph{positional}, if $f_i$ is a function mapping from $V_{0_i}$ to $V_{0_i} \uplus V_{1_i}$, i.e., the update of location depends only on the current position of local game.

\begin{defi}
A distributed game $\mathcal{G} = (\mathcal{V}_0 \uplus \mathcal{V}_1, \mathcal{E}, Acc)$ is reachability-winning by a distributed strategy $\xi = \langle f_1, \ldots, f_n\rangle$ over initial states $V_{ini} \in \mathcal{V}_0 \uplus \mathcal{V}_1$ and target states $V_{goal} \in \mathcal{V}_0 \uplus \mathcal{V}_1$,
when the following conditions hold:
\begin{itemize}
    \item $Acc = \{v_0v_1\ldots \in (\mathcal{V}_0 \uplus \mathcal{V}_1)^{\omega}\;|\;\Occ(v_0v_1\ldots) \cap V_{goal} \neq\emptyset\}$.
    \item For every play $\pi = v_0 v_1 v_2, \ldots$ where $v_0 \in V_{ini}$, player $0$ wins $\pi$ when the following constraints hold:
        \begin{itemize}
            \item $\pi \in Acc$.
            \item $\forall i \in \mathbf{N}_0.\;( v_i \in \mathcal{V}_0 \rightarrow (\forall j \in \{1,\ldots, n\}.\; ( proj(v_{i}, j) \in V_{0_j} \rightarrow proj(v_{i+1}, j) = proj(f_j(v_i), j))))$.
        \end{itemize}
\end{itemize}
\end{defi}

\section{Step A: Front-end Translation from Models to Games\label{sec.frontend.translation}}

\subsection{Step A.1: From PISEM to IM}

\begin{algorithm}[t]
\DontPrintSemicolon
\KwData{PISEM model $\mathcal{S} = (\mathcal{A}, \mathcal{N}, \mathcal{T})$}
\KwResult{Two maps $mapLB$, $mapUB$ which map from an action $\sigma$ (or a msg processing by network) to two integer arrays $lower[1\ldots n_A]$, $upper[1\ldots n_A]$}
\Begin{
    /* Initial the map for recording the lower and upper bound for action */\;
    %\textbf{let} $mapLB$, $mapUB$ = \emph{getNewMap()}\;
    \For{action $\sigma_k$ in $\mathcal{A}_i$ of $\mathcal{A}$} {
        $mapLB.put$($\sigma_k$, \textbf{new int}[1\ldots $n_A$](1)) /* Initialize to $1$ */\;
        $mapUB.put$($\sigma_k$, \textbf{new int}[1\ldots $n_A$])\;
        \lFor{$\mathcal{A}_j \in \mathcal{A}$} { $mapUB.get$($\sigma_k$)[j] := $|\overline{\sigma_j}| +2$  /* Initialize to upperbound */}\;
        $mapLB.get$($\sigma_k$)[i] = k;  $mapUB.get$($\sigma$)[i] = k+1; /* self PC */\;
    }
    \For{action $\sigma_m$ in $\mathcal{A}_i$ of $\mathcal{A}$, $m = 1,\ldots, |\overline{\sigma_i}|$} {
        \For{action $\sigma_n$ in $\mathcal{A}_j$ of $\mathcal{A}$, $n = 1,\ldots, |\overline{\sigma_j}|$ , $j \neq i$}{
            \nl\If {$\sigma_m.releaseTime > \sigma_n.deadline$}{\label{InRes1}
                $mapLB.get$($\sigma_m$)[j] := \textbf{max}\{$mapLB.get$($\sigma_m$)[j], $n + 1$\}
            }
            \nl\If {$\sigma_m.deadline < \sigma_n.releaseTime$}{\label{InRes2}
               $mapUB.get$($\sigma_m$)[j] := \textbf{min}\{$mapUB.get$($\sigma_m$)[j], $n + 1$\};
            }
        }
    }
    /* Initialize the map for recording the lower and upper bound for msg transmission */\;
    \For{action $\sigma_k = \texttt{send}(pre, ind, n, s, d, v, c)$ in $\mathcal{A}_i$ of $\mathcal{A}$} {
        $mapLB.put$($n.ind$, \textbf{new int}[1\ldots $n_A$](1))  /* Initialize to $1$ */\;
        $mapLB.get$($n.ind$)[i] := k+1 /* Strictly later than executing \texttt{send}() */\;
        $mapUB.put$($n.ind$, \textbf{new int}[1\ldots $n_A$])\;
        \lFor{$\mathcal{A}_j \in \mathcal{A}$} { $mapUB.get$($n.ind$)[j] := $|\overline{\sigma_j}| +2$  /* Initialize to upperbound */}\;
    }
    \For{action $\sigma_k = \texttt{send}(pre, ind, n, s, d, v, c)$ in $\mathcal{A}_i$ of $\mathcal{A}$} {
        %\For{action $\sigma_m$ in $\mathcal{A}_j$ of $\mathcal{A}$, $n = 1,\ldots, |\overline{\sigma_j}|$ , $j = i$, $m > k$}{
        %    \nl\If {$\sigma_k.deadline +  \mathcal{T}_n(ind) < \sigma_m.releaseTime$}{\label{InRes4}
        %       $mapUB.get$($n.ind$)[j] := \textbf{min}\{$mapUB.get$($n.ind$)[j], $m + 1$\};
        %    }
        %}
        \For{action $\sigma_m$ in $\mathcal{A}_j$ of $\mathcal{A}$, $n = 1,\ldots, |\overline{\sigma_j}|$}{ %, $j \neq i$}{
            \nl\If {$\sigma_k.releaseTime  + 0> \sigma_m.deadline$}{\label{InRes3}
                $mapLB.get$($n.ind$)[j] := \textbf{max}\{$mapLB.get$($n.ind$)[j], $m + 1$\}
            }
            \nl\If {$\sigma_k.deadline +  \mathcal{T}_n(ind) < \sigma_m.releaseTime$}{\label{InRes4}
               $mapUB.get$($n.ind$)[j] := \textbf{min}\{$mapUB.get$($n.ind$)[j], $m + 1$\};
            }
        }
    }
}
\caption{GeneratePreconditionPC\label{algo.generate.PC}}
\end{algorithm}

To translate from PISEM to IM, the key is to generate abstractions %(under-approximations)
from the release time and the deadline information specified in PISEM.
%for an action to execute. %concerning when an action is allowed to execute.
As in our formulation, the system is equipped with a globally synchronized clock, the execution of actions respecting the release time and the
deadline can be translated into a partial order.
Algorithm~\ref{algo.generate.PC} concretizes this idea
%\footnote{Here we assume that in each period, for all $\mathcal{N}_j$, each message of type $ind\in\{1,\ldots, size_j\}$ is sent at most once. In this way, the algorithm can assign an unique PC-precondition interval for every message type.}
by generating PC-intervals in all machines as
\begin{itemize}
    \item temporal preconditions for an action to execute, or
    \item temporal preconditions for a network to finish its message processing, i.e., to update a variable in the destination process with the value in the message\footnote{Here we assume that in each period, for all $\mathcal{N}_j$, each message of type $ind\in\{1,\ldots, size_j\}$ is sent at most once. In this way, the algorithm can assign an unique PC-precondition interval for every message type.}.
\end{itemize}

\begin{figure}
\centering
 \includegraphics[width=0.7\columnwidth]{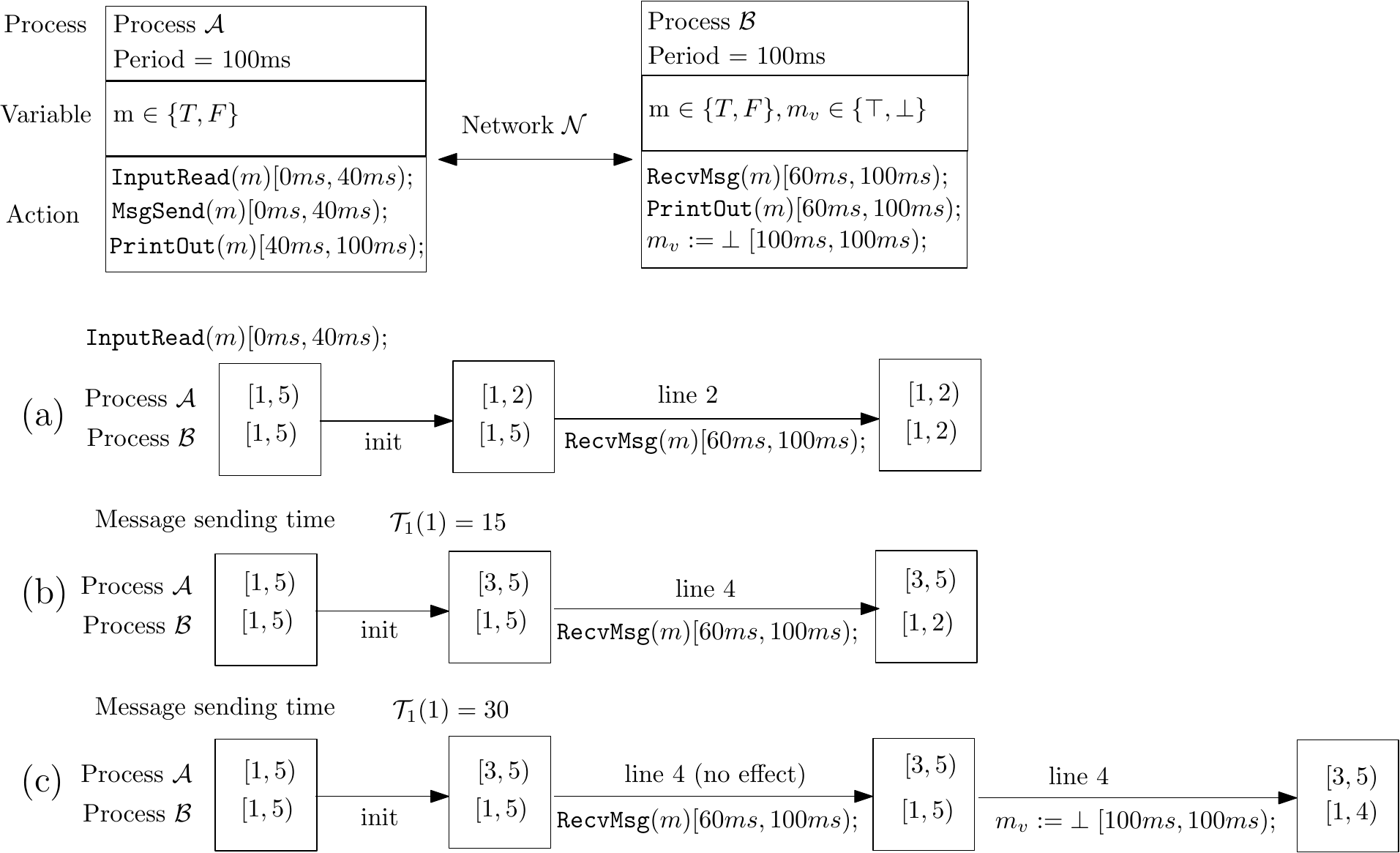}
  \caption{An illustration for Algorithm~\ref{algo.generate.PC}.}
 \label{fig:Algorithm1}
\end{figure}

%
% The algorithm above may be too simplified, without explicitly listing the case to generate its own program counter.
%

Starting from the initialization where no PC is constrained,
the algorithm performs a restriction process using four if-statements \{\textbf{(1)}, \textbf{(2)}, \textbf{(3)}, \textbf{(4)}\} listed.
\begin{itemize}
    \item In (1), if $\sigma_m.releaseTime > \sigma_n.deadline$, then before $\sigma_m$ is executed, $\sigma_n$ should have been executed.
    %\item In (2), if $\sigma_m.deadline < \sigma_n.releaseTime$, then before $\sigma_m$ is executed, $\sigma_n$ should have never been executed.
    \item In (2), if $\sigma_m.deadline < \sigma_n.releaseTime$, then $\sigma_n$ should not be executed before executing $\sigma_m$.
    %\item (3) is used to derive tighter upper bounds by considering subsequent actions after a \verb"send" action.
    \item Similar analysis is done with (3) and (4). However, we need to consider the combined effect together
            with the network transmission time: we use $0$ to represent the best case, and $\mathcal{T}_n(ind)$ for the worst case.
\end{itemize}

% [NEED SOME EXPLANATION HERE]

% Do we need to use an example to illustrate?
% The use of four constraints can be illustrated by Figure.

\noindent \textbf{[Example]}
For the example in sec.~\ref{sec.motivating.scenarios}, consider the action
\begin{small}$\sigma_1 := m \leftarrow \verb"InputRead"(v)[0, 40)$\end{small} in $\mathcal{A}_1$ of a PISEM. Algorithm~\ref{algo.generate.PC} returns $mapLB(\sigma)$ and $mapUB(\sigma)$ with two arrays $[1,1]$ and $[2,2]$, indicated in Figure~\ref{fig:Algorithm1}a.
Based on the definition of IM, $\sigma_1$ should be
executed with the temporal precondition that no action in $\mathcal{A}_2$ is executed,
satisfying the semantics originally specified in PISEM.
For the analysis of message sending time, two cases are listed in Figure~\ref{fig:Algorithm1}b and Figure~\ref{fig:Algorithm1}c, where the WCMTT is estimated as 15ms and 30ms, respectively.

\subsection{Step A.2: From IM to Distributed Game}

Here we give main concepts how a game is created after step A.1 is executed.
To create a distributed game from a given interleaving model $S_{IM} = (A, N)$, we need
to proceed with the following three steps:

\subsubsection{Step A.2.1: Creating non-deterministic timing choices for existing actions}
During the translation from a PISEM $\mathcal{S} = (\mathcal{A}, \mathcal{N}, \mathcal{T})$ to its corresponding IM $S_{IM} = (A, N)$,
for all process $\mathcal{A}_i$ in $\mathcal{A}$, for every action $\sigma [\alpha, \beta)$ where $\sigma [\alpha, \beta) \in \overline{\sigma_i}$, algorithm $1$
creates the PC-precondition interval $[\wedge_{m = 1 \ldots n_A} [pc_{m_{low}}, pc_{m_{up}})]$ of other processes. Thus in the corresponding game,
for $\sigma [\wedge_{m = 1 \ldots n_A} [pc_{m_{low}}, pc_{m_{up}})]$, each element $\sigma [\wedge_{m = 1 \ldots n_A} (pc_{m})]$,
where $pc_{m_{low}} \leq pc_{m} < pc_{m_{up}}$, is a nondeterministic transition choice which can be selected separately by the game engine.

\subsubsection{Step A.2.2: Introducing fault-tolerant choices as $\sigma_{\frac{a}{b}}$}

In our framework, fault-tolerant mechanisms are similar to actions, which consist of two parts:
\emph{action pattern} $\sigma$ and \emph{timing precondition} $[\wedge_{m = 1 \ldots n_A} [pc_{m_{low}}, pc_{m_{up}})]$.
Compared to existing actions where nondeterminism comes from timing choices, for fault-tolerance
transition choices include all combinations from (1) timing precondition and (2) action patterns available from a predefined pool.

We use the notation $\sigma_{\frac{a}{b}}$, where $\frac{a}{b} \in \mathbf{Q} \backslash \mathbf{N}$, to represent
an inserted action pattern between $\sigma_{\lfloor\frac{a}{b}\rfloor}$ and $\sigma_{\lceil\frac{a}{b}\rceil}$. With this formulation,
multiple FT mechanisms can be inserted within two consecutive actions $\sigma_i$, $\sigma_{i+1}$ originally in the system, and
the execution semantic follows what has been defined previously: as executing an action updates $\Delta_{next_i}$ to
$min\{x| x \in iSet(\overline{\sigma_i}), x > j \}$,
updating to a rational value is possible. Note that as $\sigma_{\frac{a}{b}}$ is only a fragment without temporal preconditions,
we use algorithm~\ref{algo.decide.FT.timing} to generate all possible temporal preconditions satisfying the semantics of the original interleaving model: after the synthesis only temporal conditions satisfying the acceptance condition will be chosen.

\begin{algorithm}[t]
\DontPrintSemicolon
\KwData{$\sigma_c [\wedge_{m = 1 \ldots n_A} [pc_{c,m_{low}}, pc_{c,m_{up}})]$, $\sigma_d[\wedge_{m = 1 \ldots n_A} [pc_{d,m_{low}}, pc_{d,m_{up}})]$, which are consecutive actions in $\overline{\sigma_i}$ of $A_i$ of $S_{IM} = (A, N)$,
        and one newly added action pattern $\sigma_{\frac{a}{b}}$ to be inserted between}
\KwResult{Temporal preconditions for action pattern $\sigma_{\frac{a}{b}}$: $[\wedge_{m = 1 \ldots n_A} [pc_{\frac{a}{b},m_{low}}, pc_{\frac{a}{b},m_{up}})]$}
\Begin{
    \For{$m = 1, \ldots, n_A$} {
        \eIf{ $m \neq i$} {
         $pc_{\frac{a}{b},m_{low}} := pc_{c,m_{low}}$ /* Use the lower bound of $c$ for its lower bound */\;
            $pc_{\frac{a}{b},m_{up}} := pc_{d,m_{up}}$ /* Use the upper bound of $d$ for its upper bound */\;
        }{
         $pc_{\frac{a}{b},m_{low}} := \frac{a}{b}$; $pc_{\frac{a}{b},m_{up}} := d$ %/* case when $m = i$ */
        }
    }
}
\caption{DecideInsertedFTTemplateTiming\label{algo.decide.FT.timing}}
\end{algorithm}
We conclude this step with two remarks:
\begin{list1}
    \item For all existing actions, the non-deterministic choice generation in step A.2.1 must be modified to contain
            these rational points introduced by FT mechanisms.
    \item A problem induced by FT synthesis is whether the system behavior changes due to the introduction of FT mechanisms.
            %This problem must be handled carefully.
            We answer the problem by splitting into two subproblems:
            \begin{list1}
                \item \textbf{[Problem 1]} Whether the system is still schedulable due to the introduction of FT actions,
                         as these FT actions also consume time.
                        This can only be answered when the result of synthesis is generated, and we leave this to section~\ref{sec.conversion.from.strategy.to.model}.
                \item \textbf{[Problem 2]} Whether the networking behavior remains the same. This problem \emph{must}
                        be handled before game creation, as introducing a FT message may significantly influence the worst
                        case message transmission time (WCMTT) of all existing messages, leading a
                        completely different networking behavior. The answer of this problem depends on many factors,
                        including the hardware in use, the configuration setting, and the analysis technique used for
                        the estimation of WCMTT.
                        In Appendix A we give a simple analysis for ideal CAN buses~\cite{davis:2007:controller}, which are
                        used most extensively in industrial and automotive embedded systems: in the analysis, we
                        propose conditions where newly added messages do \emph{\textbf{not}} change the existing networking behavior.
                        Similar analysis can be done with other timing-predictable networks, e.g., FlexRay~\cite{pop:2008:timing}.
            \end{list1}
\end{list1}

\subsubsection{Step A.2.3: Game Creation by Introducing Faults\label{sub.sec.from.IM.dg}}

%Step A.2.1 and A.2.2 generate a transition system where no uncontrollable vertex exists.
%An occurrence of fault can be defined as an additional transition $(s, s')$ in the transition system.
%Therefore, state $s$ in the transition system should be changed as a non-controllable vertex in the generated game. Similar changes
%should be applied on the network element, where the timing of message sending is constrained but uncontrollable.

In our implementation, we do not generate the primitive form of distributed games (DG), as the definition of DG is too primitive
to manipulate. Instead, algorithms in our implementations are based on our created variant called \emph{\textbf{symbolic distributed games} (SDG)}:

\begin{defi}
Define a symbolic distributed game \begin{small}$\mathcal{G}_{ABS} = (V_{f}\uplus  V_{CTR}\uplus V_{ENV}, A, N, \sigma_{f}, pred)$\end{small}.
\begin{list1}
    %\item $F = ( \sigma_{f}, pred_{fault})$ is the environment (fault) process.
    %\begin{list1}
    \item $V_{f}$, $V_{CTR}$, $V_{ENV}$ are disjoint sets of (fault, control, environment) variables.
    \item $pred: V_{f}\times V_{CTR}\times V_{ENV} \rightarrow \{\verb"true",\verb"false"\}$ is the
            partition condition.
    %\end{list1}
      \item $A = \bigcup_{i = 1\ldots n_A} A_i$ is the set of \textbf{symbolic local games (processes)} ,
            where in $A_{i} = (V_{i} \cup V_{env_i} , \overline{\sigma_i})$,
  %\item $A = \bigcup_{i = 1\ldots n_A} A_i$ is the set of control processes, where in $A_{i} = (V_{i} \cup V_{env_i} , \overline{\sigma_i})$,
    \begin{itemize}
         \item $V_i$ is the set of variables, and $V_{env_i} \subseteq V_{ENV}$.
        \item  $\overline{\sigma_i} := \bigcup\sigma_{i_1} \langle \wedge_{m = 1,\ldots, n_A} pc_{i_{1_{m}}}\rangle; \ldots;
            \bigcup \sigma_{i_k}\langle \wedge_{m = 1,\ldots, n_A} pc_{i_{k_{m}}}\rangle$ is a sequence, where \\$\forall j = 1, \ldots, k$, $\bigcup\sigma_{i_j} \langle \wedge_{m = 1,\ldots, n_A} pc_{i_{j_{m}}}\rangle$ is a set of choice actions for player-0 in $A_i$.
            \begin{itemize}
                \item $\sigma_{i_j}$ is defined similarly as in IM.
               % $\sigma_{i_j} := \verb"send"(pre, index, n, d, v, c)\; | \;\verb"receive"(pre, s, c)\; | \;a \leftarrow \verb"e"  $ is an atomic action, where $a, c, e, pre, v, n, s, d$ are defined similarly as in IM.
                \item \begin{small}$\forall m = \{1, \ldots, n_A\}$, $ pc_{i_{j_{m}}} \in [pc_{i_j,m_{low}}, pc_{i_j,m_{up}}), pc_{i_j,m_{low}}, pc_{i_j,m_{up}} \in iSet(\overline{\sigma_m})$\end{small}.
            \end{itemize}
        \item $V_{CTR} = \bigcup_{i= 1\ldots n_A} V_i$.
    \end{itemize}
   \item $N = \bigcup_{i = 1\ldots n_N} N_i$, $N_i = (T_i, size_i, tran_i)$ is the set of network processes.
    \begin{itemize}
        \item $T_i$ and $size_i$ are defined similarly as in IM.
        %$: \mathbf{N} \rightarrow \bigwedge_{m = 1 \ldots n_A} (\{1, \ldots, |\overline{\sigma_m}|+2 \}, \{1, \ldots, |\overline{\sigma_m}|+2 \})$ is a function which maps the index (or priority) of a message to the PC-precondition interval of other processes.
        %\item $size_i$ is the number of messages used in $\mathcal{N}_i$.
         \item \begin{small}$tran_i:  V_{f}  \times (\{\verb"true", \verb"false"\} \times \{1, \ldots, n_A\}^{2}\times \bigcup_{i = 1,\ldots, n_A} (V_{i} \cup V_{env_i}) \times \mathbf{Z} \times \{1, \ldots ,size_i\}) \rightarrow V_{f}  \times (\{\verb"true", \verb"false"\} \times \{1, \ldots, n_A\}^{2}\times \bigcup_{i = 1,\ldots, n_A} (V_{i} \cup V_{env_i}) \times \mathbf{Z} \times \{1, \ldots ,size_i\})$\end{small} is the network transition relation for processing messages (see sec.~\ref{sub.sec.PISEM} for meaning), but can be influenced by additional variables in $V_{f}$.
        %\item \begin{small}$tran_i: V_{f}  \times \{\verb"true", \verb"false"\} \times \{1, \ldots, n_A\}^{2}\times V \times \mathbf{R} \times \{1, \ldots ,size_i\} \rightarrow V_{f} \times
        %\{\verb"true", \verb"false"\} \times \{1, \ldots, n_A\}^{2}\times \mathbf{Z} \times \mathbf{R} \times \{1, \ldots ,size_i\}$\end{small} is the network transition relation for processing the message (see sec.~\ref{sub.sec.PISEM} for meaning), but can be influenced by variables in $V_{f}$.
    \end{itemize}
       \item \begin{small}$\sigma_{f}: V_{f}\times V_{CTR}\times V_{ENV}\times \bigwedge_{i = 1 \ldots n_A} iSet(\overline{\sigma_i}) \rightarrow V_{ENV} \times V_{f} \times \bigwedge_{i = 1 \ldots n_A} iSet(\overline{\sigma_i})$\end{small} is the environment update relation.
\end{list1}
\end{defi}

\begin{figure}[t]
\centering
\begin{tabular}[t]{|l|l|l|}
\hline
 & DG & SDG \\
\hline
State space & product of all vertices  & product of all variables  \\
             &   in local games         & (including variables used in local games) \\
\hline
Vertex partition ($V_0$ and $V_1$) & explicit partition  & use $pred$ to perform partition \\
\hline
Player-0 transitions & defined in local games  & defined in $\overline{\sigma_i}$ of $A_i$, for all $i \in \{1,\ldots,n_A\}$ \\
\hline
Player-1 transitions & explicitly specified  & defined in $N$ and $\sigma_{f}$ \\
                     & in the global game    &  \\

\hline
\end{tabular}
\caption{Comparison between DG and SDG}
 \label{table:comparison}
\end{figure}

We establish an analogy between SDG and DG using Figure~\ref{table:comparison}.
\begin{enumerate}
    \item The configuration $v$ of a SDG is defined as the product of all variables used.
    \item A play for a SDG starting from state $v_0$ is a maximal path $\pi=v_0 v_1 \ldots$, where
        \begin{list1}
            \item In $v_k$, player-1 determines the move $(v_k,v_{k+1})\in E$ when $pred(v_k)$ is evaluated to \verb"true" (\verb"false" for player-0); the partition of vertices $V_0$ and $V_1$ in a SDG is implicitly defined based on this, rather than
                specified explicitly as in a distributed game.
            \item A move $(v_k,v_{k+1})$ is a selection of executable transitions defined in $N$, $\sigma_{f}$, or $A$; in our formulation, transitions in $N$ and $\sigma_{f}$ are all environment moves\footnote{As the definition of distributed games features multiple processes having no interactions among themselves but only with the environment, a SDG is also a distributed game. In the following section, our proof of results and algorithms are all based on DG.}, while transitions in $A$ are control moves\footnote{
                    This constraint can be released such that transitions in $A$ can either be control (normal) or environment (induced by faults) moves; here
                    we leave the formulation as future work.}.
    \end{list1}
    \item Lastly, a distributed positional strategy for player-0 in a SDG can be defined analogously as to uniquely select an action from the set $\bigcup \sigma_{\alpha_j} \langle \wedge_{m = 1,\ldots, n_A},  pc_{\alpha_{j_{m}}}\rangle$, for all $A_i$ and for all program counter $j$ defined in $\overline{\sigma_i}$. Each strategy should be insensitive of contents in other symbolic local games.
\end{enumerate}

%\begin{table}

%\end{table}

 %in $A, N, F$,
%and for convenience, we use $proj(v, F)$ to represent the configuration of the variables in the fault process.

%With the definition of SDG in mind,

We now summarize the logical flow of game creation using Figure~\ref{fig:SDG.creation}.
\begin{list1}
    \item (a) Based on the fixed number of slots (for FT mechanisms) specified by the user, extend $IM$ to $IM_{frac}$ to contain fractional PC-values induced by the slot.
    \item (b) Create $IM_{frac+FT}$, including the sequence of choice actions (as specified in the SDG) by
            \begin{list1}
                \item Extracting action sequences defined in $IM_{frac}$ to choices (step A.2.1).
                \item Inserting FT choices (step A.2.2).
            \end{list1}
    \item (c) Introduce faults and partition player-0 and player-1 vertices:
            In engineering, a \emph{fault model} specifies potential undesired behavior of a piece of equipment, such that engineers can predict the consequences of system behavior. % induced by this fault.
             Thus, a \emph{fault} can be formulated with three tuples\footnote{For complete formulation of
             fault models, we refer readers to our earlier work~\cite{Cheng:2009:Timing}.}:
              \begin{enumerate}
                  \item The fault type (an unique identifier, e.g., \verb"MsgLoss", \verb"SensorError").
                  \item The maximum number of occurrences in each period.
                  \item Additional transitions not included in the original specification of the system (\emph{fault effects}).
             \end{enumerate}
            We perform the translation into a game using the following steps.
                \begin{list1}
                    \item For (1), introduce variables to control the triggering of faults.
                    \item For (2), introduce counters to constrain the maximum number of fault occurrences in each period.
                    \item For (3), for each transition used in the component influenced by the fault, create a corresponding
                                 fault transition which is triggered by the variable and the counter; similarly create a transition
                                 with normal behavior (also triggered by the variable and the counter). Notice that our
                                 framework is able to model faults actuating on the FT mechanisms, for instance, the behavior
                                 of network loss on the newly introduced FT messages.
                \end{list1}
            %\item An example on how partition of player-0 and player-1 vertices is don.
                %as an occurrence fault can be viewed as an additional edge of the transition system,
                %we need to convert the textural statement and argument the
\end{list1}

\begin{figure}[t]
\centering
 \includegraphics[width=0.7\columnwidth]{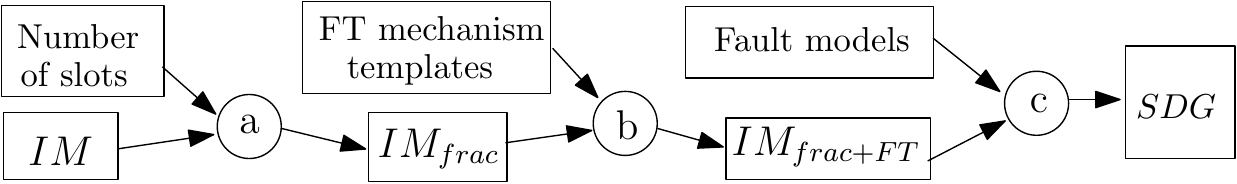}
  \caption{Creating the SDG from IM, FT mechanisms, and faults.}
 \label{fig:SDG.creation}
\end{figure}

\noindent \textbf{[Example]}
We outline how a game (focusing on fault modeling) is created with the example in sec.~\ref{sec.motivating.scenarios}; similar approaches can be applied for input errors or message corruption; here the modeling of input (for \verb"InputRead(m)") is skipped.
\begin{list1}
    %\item In the network, let the boolean variable $occu$ (representing the network occupance) to be contained in the set $V_{ENV}$.
    \item Create the predicate $pred$: $pred$ is evaluated to \verb"false" in all cases except (a) when the boolean variable $occu$ (representing the
            network occupance) is evaluated to $\verb"true"$ and (b) when for all $i \in \{1,\ldots, n_A\}$, $\Delta_{next_i} = |\overline{\sigma_i}| + 1$ (end of period); the predicate partitions player-0 and player-1 vertices.
            % In our implementation, this corresponds to two variables "End_of_Cycle" and "_forceTransition"
    \item For all process $i$ and  program counter $j$, the set of choice actions $\bigcup \sigma_{\alpha_j} \langle \wedge_{m = 1,\ldots, n_A},  pc_{\alpha_{j_{m}}}\rangle$ are generated based
            on the approach described previously.
    \item Create variable $v_f \in V_{f}$, which is used to indicate whether the fault (\texttt{MsgLoss}) has been activated in this period.
    \item In this example, as the maximum number of fault occurrences in each period is $1$, we do not need to create additional counters.
            % In our implementation, this corresponds to the variable "_faultActivated".
    %\item Create boolean variable $c \in V_{f}$, which is used to govern the number of message loss in each period (in our example, at lost one can occur).
    \item For each message sending transition $t$ in the network, create two normal transitions $(v_f = \verb"true" \wedge v'_f = \verb"true")  \wedge t$ and
            $(v_f = \verb"false" \wedge v'_f = \verb"false")  \wedge t$ in the game.
            %; these transitions are normal transitions.
    \item For each message sending transition $t$ in the network, generate a transition $t'$ where the message is sent, but the value is not
            updated in the destination. Create a fault transition $(v_f = \verb"false" \wedge v'_f = \verb"true") \wedge t'$ in the game. %;this transition is a fault transition.
    \item Define $\sigma_{f}$ to control $v_f$: if for all $i \in \{1,\ldots, n_A\}$, $\Delta_{next_i} = |\overline{\sigma_i}| + 1$, then
            update $v_f$ to \verb"false" as $\Delta_{next_i}$ updates to $1$ (reset the fault counter at the end of the period).
            %\item Otherwise, keep $v_f$ to be the same value during the update of any $\Delta_{next_i}$.
\end{list1}

%With the formulation, generating games in sec.~\ref{sub.sec.from.IM.dg}
%is straightforward, and all algorithms for solving distributed games can be easily translated to solve symbolic distributed games.

%We refer readers to Appendix B for detailed formulation.

%We conclude this section with one remark: as FT mechanisms are created earlier than the introduction of fault models,
%our formulation allows to model faults actuating
%on the introduced FT mechanisms, which is reasonable in practice.
%For instance, sending a new message for the use of fault-tolerance might also fail under an unreliable network.

% Here we give two general remarks:
%\begin{itemize}
%    \item As FT mechanisms are created earlier than the introduction of fault models, our formulation allows to model faults actuating
%            on the FT mechanisms introduced. For instance, sending a new message for the use of fault-tolerance
%            might also fail under an unreliable network.
%    \item \textbf{(Generate the partition for vertices)  Write it later when the whole formulation in the appendix is ready [Chih-Hong].}
%\end{itemize}

% To generate a distributed game, the following transformation
% is executed:
% For complete formulation, we refer readers to the appendix.
% encode the fault-tolerant mechanisms and faults.
% 1. Introduced FT mechanisms can also be faulty.
% 2. Vertex partition.

\section{Step B: Solving Distributed Games\label{sec.game.solving}}

%\section{Practical Aspects in Solving Distributed Games for Reachability Conditions}

We summarize the result from~\cite{mohalik:2003:distributed} as a general property of distributed games.
\begin{theo}
There exists distributed games with global winning strategy but (a) without distributed memoryless strategies, or (b) all distributed
strategies require memory. In general, for a finite distributed game, it is undecidable to check whether a distributed
strategy exists from a given position~\cite{mohalik:2003:distributed}.
\end{theo}

As the problem is undecidable in general, we restrict our interest in finding a distributed
positional strategy for player $0$, if there exists one.
%Finding a positional strategy can be useful, as a strategy immediately implies a synthesized controller.
We also focus on games with reachability winning conditions. By posing the restriction, the problem is NP-Complete.

\begin{theo} \textbf{[$PositionalDG_0$]}
Given a distributed game $\mathcal{G} = (\mathcal{V}_0 \uplus \mathcal{V}_1, \mathcal{E})$, an initial state $x=(x_1, \ldots, x_n)$ and
a target state $t=(t_1, \ldots, t_n)$, deciding whether there exists a positional (memoryless) distributed strategy for player-$0$  from $x$ to $t$ is NP-Complete.
\end{theo}

\begin{proof}

We first start by recalling the definition of attractor, a term which is commonly used in the game and later applied in the proof.
Given a game graph $G=(V_0 \uplus V_1, E)$, for $i\in\{0,1\}$ and $X\subseteq V$, the map $\attr_i(X)$ is defined by
\[
 \attr_i(X) := X \cup \{ v\in V_i \mid vE \cap X \neq \emptyset \} \cup \{ v\in V_{1-i} \mid \emptyset \neq vE \subseteq X \},
\]
i.e., $\attr_i(X)$ extends $X$ by all those nodes from which either player $i$ can move to $X$ within one step or player $1-i$
cannot prevent to move within the next step. ($vE$ denotes the set of successors of $v$.)
Then $\Attr_i(X) := \bigcup_{k\in\N} \attr^k_i(X)$ contains all nodes from which player $i$ can force any play to visit the set $X$.

\begin{figure}
\centering
 \includegraphics[width=0.8\columnwidth]{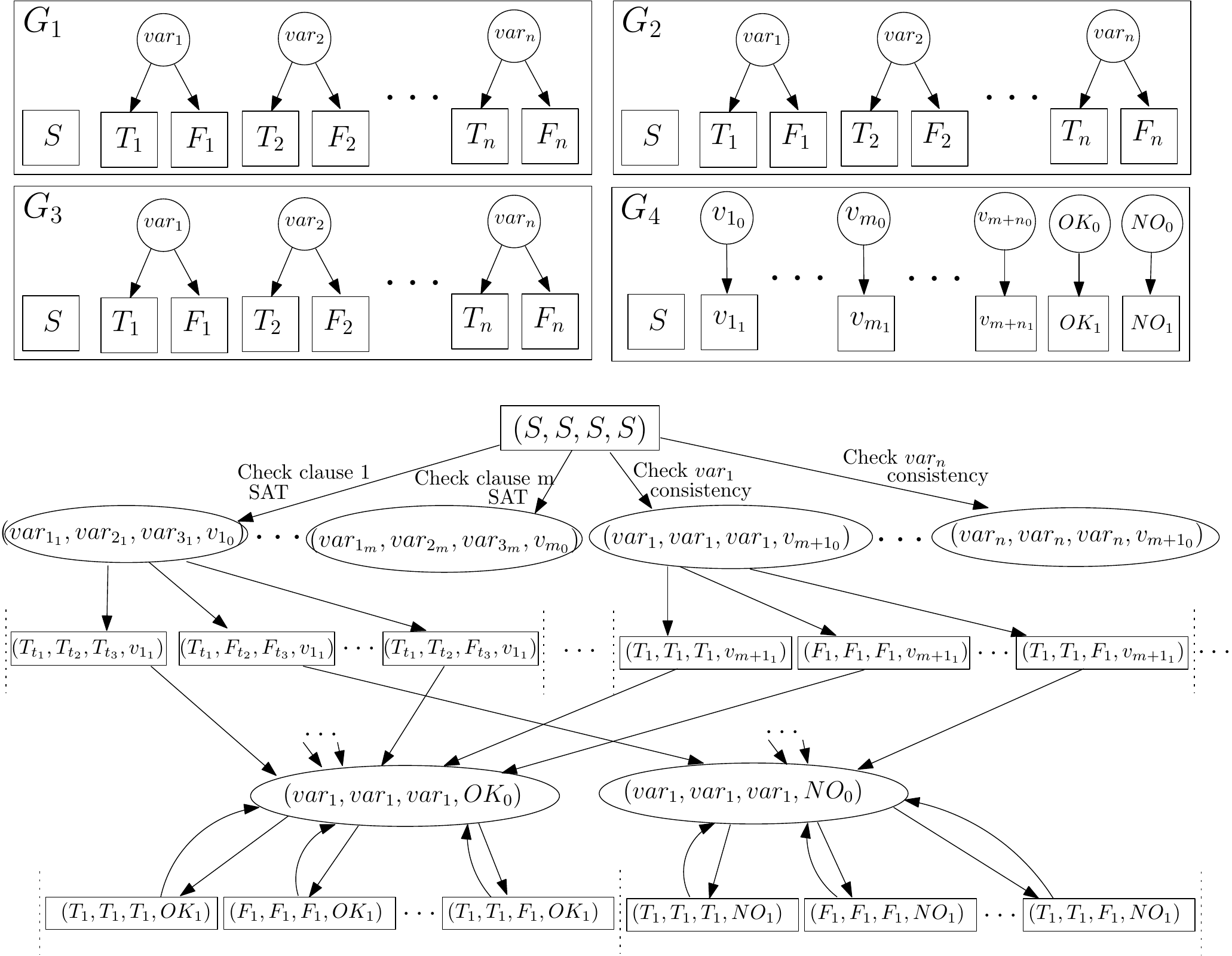}
  \caption{Illustrations for the reduction from 3SAT to $PositionalDG_0$.}
 \label{fig:Proof}
\end{figure}

We continue our argument as follows.

\textbf{[NP] }The reachability problem for a distributed game can be solved in NP: a solution instance
            $\xi = \langle f_1, \ldots, f_n\rangle$ is a strategy which
        selects exactly one edge for every control vertex in the local game. As the distributed game graph is known, after the selection we calculate the reachability attractor $\Attr_0(\{t\})$ of the distributed game: during the calculation we overlook
        transitions which is not
        selected (in the strategy) in the local game. This means that in the distributed game, to add a
        control vertex $v\in \mathcal{V}_0 $ to the attractor using the edge $(v,u)$, we must
        ensure that $\forall j \in \{1,\ldots, n\}.\; ( proj(v_{i}, j) \in V_{0_j} \rightarrow proj(u, j) = proj(f_j(v), j))$.
        Lastly, we check if the initial state is contained; the whole calculation and checking process can be done in deterministic P-time.

\textbf{[NP-C]} For completeness proof, we perform a reduction from 3SAT to the finding of positional strategies in a distributed game.
          Given a set of 3CNF clauses $\{C_1, \ldots, C_m\}$ under the set of literals $\{var_1, \overline{var_1}, \ldots, var_n, \overline{var_n}\}$ and
          variables $\{var_1, \ldots, var_n\}$,
          the distributed game $\mathcal{G}$ is created as follows (see Figure~\ref{fig:Proof} for illustration):
    \begin{list1}
        \item Create 3 local games $G_1$, $G_2$, and $G_3$, where for $G_i=(V_{0_i} \uplus V_{1_i},E_i)$:
            \begin{list1}
                \item $V_{0_i} = \{var_1, \ldots, var_n \}$, $V_{1_i} = \{S, T_{var_1}, F_{var_1}, \ldots, T_{var_n}, F_{var_n}\}$.
                \item $E_i = \bigcup_{j = 1,\ldots, n} \{(var_j, T_{var_j}), (var_j, F_{var_j})\}$.
          \end{list1}
        \item Create local game $G_4=(V_{0_4} \uplus V_{1_4},E_4)$:
         \begin{list1}
                \item $V_{0_4} = \{OK_0, NO_0\}\;\cup\;\bigcup_{j = 1, \ldots, m+n}\{v_{j_0}\}$.
                \item $V_{1_4} = \{S, OK_1, NO_1\}\;\cup\;\bigcup_{j = 1, \ldots, m+n}\{v_{j_1}\}$.
                \item $E_4 = \bigcup_{j = 1, \ldots, m+n}\{(v_{j_0}, v_{j_1})\}\;\cup\;\{(OK_0, OK_1), (NO_0, NO_1)\}$.
          \end{list1}
        \item Second, create the distributed game $\mathcal{G}$ from local
                games above, and define the set of environment transition to include the following types using the 3SAT problem:
            \begin{enumerate}
                \item (Intention to check SAT) In the 3SAT problem, for clause $C_i = (l_{1_i}\vee l_{2_i} \vee l_{3_i})$, let the variable
                        for literals $l_{1_i},l_{2_i}, l_{3_i}$ be $var_{1_i},var_{2_i}, var_{3_i}$.
                        Create a transition in the distributed game from $(S,S,S,S)$ to $(var_{1_i},var_{2_i}, var_{3_i}, v_{i_0})$.
                 \item (Intention to check consistency) In the 3SAT problem, for variable $var_i$,
                        Create a transition in the distributed game from $(S,S,S,S)$ to $(var_i, var_i, var_i, v_{m+i_0})$.
                 \item (Result of clause) In the 3SAT problem, for clause $C_i = (l_{1_i}\vee l_{2_i} \vee l_{3_i})$, let the variable for the clause
                        be $var_{1_i},var_{2_i}, var_{3_i}$. We refer the vertex evaluating $var_{j_i}$ as \verb"true" to $T_i$ in the local game $G_j$; similarly, we use $F_i$ for a variable being evaluated \verb"false".
                        For each clause $C_i$, enumerate over 8 cases for the assignments of $var_{1_i},var_{2_i}, var_{3_i}$ which make $C_i$ true.
                        \begin{enumerate}
                            \item For cases which makes the assignment true, create an edge from the assignment to $(var_1, var_1, var_1, OK_0)$;
                            for example, if \begin{small}$var_{1_i} = \verb"true", var_{2_i} = \verb"false", var_{3_i} = \verb"true"$\end{small} makes a satisfying assignment to $C_i$, create
                                an edge $((T_{1_i}, F_{2_i}, T_{3_i}, v_{i_1}), (var_1, var_1, var_1, OK_0))$.
                            \item For cases which makes the assignment false, create an edge from the assignment to $(var_1, var_1, var_1, NO_0)$.
                        \end{enumerate}
                 \item (Result of variable consistency) For all $i \in \{1, \ldots, n\}$:
                    \begin{enumerate}
                        \item Create two edges $((T_i, T_i, T_i, v_{m+i_1}), (var_1, var_1, var_1, OK_0))$ and \\
                                $((F_i, F_i, F_i, v_{m+i_1}), (var_1, var_1, var_1, OK_0))$.
                        \item For other 6 combinations $(T_i, F_i, F_i, v_{m+i_1}), (F_i, F_i, T_i, v_{m+i_1}),\\ (T_i, F_i, F_i, v_{m+i_1}), (T_i, T_i, F_i, v_{m+i_1}), (F_i, T_i, T_i, v_{m+i_1}),\\ (T_i, F_i, T_i, v_{m+i_1})$, create edges to $(var_1, var_1, var_1, NO_0)$.
                    \end{enumerate}
                 \item (Continuous execution) For all $i \in \{1, \ldots, n\}$:
                       \begin{enumerate}
                        \item For all combinations $(T_i, F_i, F_i, OK_1), (F_i, F_i, T_i, OK_1), (T_i, F_i, F_i, OK_1), \\(F_i, F_i, T_i, OK_1), (T_i, F_i, F_i, OK_1), (T_i, T_i, F_i, OK_1), (F_i, T_i, T_i, OK_1),\\ (T_i, F_i, T_i, OK_1)$, create edges to $(var_1, var_1, var_1, OK_0)$.
                         \item For all combinations $(T_i, F_i, F_i, NO_1), (F_i, F_i, T_i, NO_1), (T_i, F_i, F_i, NO_1), \\(F_i, F_i, T_i, NO_1), (T_i, F_i, F_i, NO_1), (T_i, T_i, F_i, NO_1), (F_i, T_i, T_i, NO_1),\\ (T_i, F_i, T_i, NO_1)$, create edges to $(var_1, var_1, var_1, NO_0)$.
                    \end{enumerate}
            \end{enumerate}
    \end{list1}

We claim that $\{C_1, \ldots, C_m\}$ is satisfiable iff $\mathcal{G}$ has a positional distributed strategy to reach $(var_1, var_1, var_1, OK_0)$
from $(S, S, S, S)$.
\begin{enumerate}
\item If $\{C_1, \ldots, C_m\}$ is satisfiable, let the set of satisfying literals be $L'$, and assume that for all literals, in each pair ($var_i, \overline{var_i}$) exactly one of them is in $L'$ (this is always possible).
    For the distributed game $\mathcal{G}$, in local games $G_1$, $G_2$ and
    $G_3$, let the positional strategy for control vertex $var_i$ move to $T_i$ if $var_i\in L'$, and move to $F_i$ if $\overline{var_i}\in L'$ (for $G_4$, simply use the local edge).
    In a play, as player-1 starts the move, any of his selection leads to a player-0 vertex:
    \begin{list1}
        \item If player-1 choose edges of type 1 (intension to check the clause of SAT), for $G_1$, $G_2$ and $G_3$, the vertex uses its
                positional strategy, which corresponds to the assignment in the clause. The combined move then forces player-1 to choose an edge of type 3(a),
                 leading to the target state.
        \item If player-1 choose edges of type 2 (intension to check the consistency), as the positional strategies for $G_1$, $G_2$ and $G_3$ are all derived from the same satisfying instance of the 3SAT problem, for each strategy, it performs the same move from $var_i$ to $T_i$ or to $F_i$; the combined move of player-0 forces player-1 to choose an edge of type 4(a),
            leading to the target state.
    \end{list1}
\item Consider a distributed positional strategy $\langle f_1, f_2, f_3, f_4\rangle$ which reaches\\ $(var_1, var_1, var_1, OK_0)$
        from $(S, S, S, S)$. In $G_1$, for each control vertex $var_i$, it points to $T_i$ or $F_i$. The positional strategy of $G_1$ generates a satisfying instance of the 3SAT problem:
        \begin{list1}
            \item Assign $var_i$ in the 3SAT problem to \verb"true" if the strategy points vertex $var_i$ in $G_1$ to $T_i$.
            \item Assign $var_i$ in the 3SAT problem to \verb"false" if the strategy points vertex $var_i$ in $G_1$ to $F_i$.
        \end{list1}
        %The assignment satisfies all clauses, as .
\end{enumerate}

We analyze the size of the game and the time required to perform the reduction.
\begin{list1}
    \item For $i = 1, 2, 3$, $G_i$ contains $3n+1$ vertices, and $G_4$ has $2(m+n+2)+1$ vertices. As the total vertices of the distributed game is the product,
        it is polynomial to the original 3SAT problem instance.
    \item Consider the time required to perform reduction from 3SAT to $PositionalDG_0$:
    \begin{list1}
        \item For $i = 1,2,3$, $G_i$, they are constructed in $\mathcal{O}(n)$.
        \item $G_4$ is constructed in $\mathcal{O}(m+n)$.
        \item For the distributed game, vertices are constructed polynomial to $m$ and $n$, more precisely $\mathcal{O}(n^{3}(m+n))$.
        \item For edges in the distributed game, we consider the most complicated case, i.e. creating an edge of type 3. Yet it takes constant time to check
            and establish the connection, and for each player-1 vertex except $(S,S,S,S)$ which has $m+n$ edges, at most 8 edges are created. Therefore, the total required time for edge construction is also polynomial to $m$ and $n$.
    \end{list1}
\end{list1}
Therefore, 3SAT $\leq_{poly} PositioalDG_0$, which concludes the proof.
%\qed
\end{proof}

%With the NP-completeness proof, finding a distributed reachability strategy amounts to the process of searching.
%For searching, we consider (a) bounded search which combines the nodes in the search tree with BDDs, and (b) distributed version of the
%witness algorithm using SAT unrolling. Details can be found in the extended version.

With the NP-completeness proof, finding a distributed reachability strategy for distributed games amounts to the process of searching.
For example, it is possible to perform a bounded-depth forward search over choices of local transitions:
during the search, the selection of edges is constructed as a tree node in the search tree, and the set of reachable vertices (represented as BDD) based on the selection is also stored in the tree node. This method is currently implemented in our framework.

%We consider two algorithms to find the positional strategy.
%\begin{enumerate}
%    \item The first is based on a bounded-depth forward search over choices of local transitions: during the search,
%            the selection of edges is constructed as a tree node in the search tree, and the set of reachable vertices (represented as BDD) based on the selection is also stored in the tree node. %Details are omitted.
    %\item We also consider a P-time approximation algorithm; details are in Appendix D.
    %\item

\subsection{Solving Distributed Games using SAT Methods}
Apart from the search method above, in this section we give an alternative approach based on a reduction to SAT.  Madhusudan, Nam, and Alur~\cite{alur:2005:symbolic} designed the \emph{bounded witness algorithm} (based on unrolling) for solving reachability (local) games. Although based on their experiment, the witness algorihm is not as efficient as the BDD based approach in centralized games,
we find this concept potentially useful for solving distributed games. For this, we have created a variation (Algorithm~\ref{algo.SAT.DG0})
for this purpose. %Here for the ease of illustration, we use unary encoding for vertices in the game.

To provide an intuition, first we paraphrase the concept of witness defined in~\cite{alur:2005:symbolic}, a set of
states which witnesses the fact that player 0 wins.
In~\cite{alur:2005:symbolic}, consider the generated SAT problem from a local game $G=(V_0 \uplus V_1,E)$ trying to reach from $V_{init}$ to $V_{goal}$: for $i = 1, \ldots, d$ and vertex $v \in V_0 \uplus V_1$, variable $\langle v\rangle_i =\texttt{true}$ when one of the following holds:
\begin{enumerate}
    \item $v \in V_{init}$ and $i = 1$ (if $v \not\in V_{init} \wedge i = 1$ then $\langle v\rangle_i =\texttt{false}$).
    \item $v \in V_{goal}$ (if $v \not\in V_{goal} \wedge i = d$ then $\langle v\rangle_i =\texttt{false}$).
    \item $v \in V_0 \setminus V_{goal}$ and $\exists v' \in V_0 \uplus V_1.\; \exists e \in E. \; \exists j > i.\; ( e=(v, v') \wedge \langle v'\rangle_j = \texttt{true})$
    \item $v \in V_1 \setminus V_{goal}$ and $\forall e = (v, v') \in E. \; \exists j > i.\;  \langle v'\rangle_j = \texttt{true}$
\end{enumerate}

%For the ease of understanding, first we paraphrase the concept of witness defined in~\cite{alur:2005:symbolic}, a set of
%states which witnesses the fact that player 0 wins. For a generated boolean constraint system (SAT problem) where, for $i = 1, \ldots, d$ and vertex $v = (v_1, \ldots, v_m) \in \mathcal{V}_0 \uplus \mathcal{V}_1$, $\langle v\rangle_i =\texttt{true}$ when one of the following holds:
%\begin{enumerate}
%    \item $v \in V_{init}$ and $i = 1$
%        \begin{enumerate}
%            \item If $v \not\in V_{init} \wedge i = 1$ then $\langle v\rangle_i =\texttt{false}$.
%        \end{enumerate}
%    \item $v \in V_{goal}$
%        \begin{enumerate}
%            \item If $v \not\in V_{goal} \wedge i = d$ then $\langle v\rangle_i =\texttt{false}$.
%        \end{enumerate}
%    \item $v \in \mathcal{V}_0 \setminus V_{goal}$ and $\exists v' \in \mathcal{V}_0 \uplus \mathcal{V}_1.\; \exists e \in \mathcal{E}. \; \exists j > i.\; ( e=(v, v') \wedge \langle v'\rangle_j = \texttt{true})$
%    \item $v \in \mathcal{V}_1 \setminus V_{goal}$ and $\forall e = (v, v') \in \mathcal{E}. \; \exists j > i.\;  \langle v'\rangle_j = \texttt{true})$
%\end{enumerate}

This recursive definition implies that if $v$ in $V_0$ (resp. in $V_1$) is not the goal but in the witness set, then exists one
(resp. for all) successor $v'$ which should either be (i) in a goal state or (ii) also in the witness: note that for (ii), the
number of allowable steps to reach the goal is decreased by one.
This definition ensures that all plays defined in the witness reaches the goal from the initial state within $d-1$ steps:
 If a play (starting from initial state) has proceeded $d-1$ steps and reached $u\not\in V_{goal}$, then based on (2),
 $\langle u\rangle_d$ should be \texttt{false}. However, based on (1), (3), (4) the $\langle u\rangle_d$ should be set to \texttt{true} (reachable from initial states using $d-1$ steps). Thus the SAT problem should be unsatisfiable.

In general, Algorithm~\ref{algo.SAT.DG0} creates constraints based on the above concept, but compared to the bounded
local game reachability algorithm in~\cite{alur:2005:symbolic}, it contains slight modifications:
\begin{enumerate}
    \item When a variable $\langle v\rangle_i$ is evaluated to \verb"true", it means that vertex $v$ can reach the target state within $d-i$ steps, which
            is the same as what is defined in~\cite{alur:2005:symbolic}. However, we introduce more variables for edges in local
            games, which is shown in STEP 1: when a variable $\langle e\rangle$ is evaluated to \verb"true", the distributed strategy uses the local transition $e$.
    \item To achieve locality, we must include constraints specified in STEP $4$: the positional (memoryless) strategy disallows to change the use of local edges from a given vertex.
    \item We modify the impact of control edge selection in STEP $6$ by adding an additional implication "$\langle e\rangle \Rightarrow$"
        over the original constraint in the witness algorithm~\cite{alur:2005:symbolic}. Here as in Mohalik and Walukiwitz's formulation,
        all subgames in a control position should proceed a move (the progress of a global move is a combination of local moves), we need to create constraints considering all possible local edge combinations.
\end{enumerate}

\begin{algorithm}
\DontPrintSemicolon
\KwData{Distributed game graph $\mathcal{G} = (\mathcal{V}_0 \uplus \mathcal{V}_1, \mathcal{E})$, set of initial states $V_{init}$,
    set of target states $V_{goal}$, the unrolling depth $d$}
\KwResult{Output: whether a distributed positional strategy exists to reach $V_{goal}$ from $v_{init}$}
\Begin{
    \textbf{let} clauseList := \emph{getEmptyList()} /* Store all clauses for SAT solvers */\;
    /* STEP 1: Variable creation */\;
    \For{$v = (v_1, \ldots, v_m) \in \mathcal{V}_0 \uplus \mathcal{V}_1$} {
        \textbf{create} $d$ boolean variables $\langle v_1, \ldots, v_m\rangle_1, \ldots, \langle v_1, \ldots, v_m\rangle_d$;
    }
    \For{local control transition $e = (x_i, x_i') \in E_i, x_i \in V_{0_i}$} {
        \textbf{create} boolean variable $\langle e\rangle$;
    }
    /* STEP 2: Initial state constraints */\;
    \For{$v = (v_1, \ldots, v_m) \in \mathcal{V}_0 \uplus \mathcal{V}_1$} {
        \eIf{$(v_1, \ldots, v_m) \in V_{init}$} {
            clauseList.add([$\langle v_1, \ldots, v_m\rangle_1$])\;
        }{
            clauseList.add([$\neg\langle v_1, \ldots, v_m\rangle_1$])\;
        }
    }

    /* STEP 3: Target state constraints */\;
    \For{$v = (v_1, \ldots, v_m) \in \mathcal{V}_0 \uplus \mathcal{V}_1$} {
        \eIf{$(v_1, \ldots, v_m) \in V_{goal}$} {
            clauseList.add([$\langle v_1, \ldots, v_m\rangle_1 \wedge\ldots \wedge \langle v_1, \ldots, v_m\rangle_d$])\;
        }{
            clauseList.add([$\neg \langle v_1, \ldots, v_m\rangle_d$])\;
        }
    }

    /* STEP 4: Unique selection of local transitions (for distributed positional strategy) */\;
    \For{ local control transition $e = (x_i, x_i') \in E_i, x_i \in V_{0_i}$} {
        \For{local transition $e_1 = (x_i, x_{i_1}'), \ldots, e_k = (x_i, x_{i_k}') \in E_i, e_1 \ldots e_k \neq e$} {
            clauseList.add([$\langle e\rangle \Rightarrow (\neg \langle e_1\rangle \wedge \ldots \wedge \neg \langle e_k\rangle )$])\;
        }
    }

    /* STEP 5: If a control vertex is in the attractor (winning region) but not a goal,\;
    an edge should be selected to reach the goal state */\;
    \For{$v = (v_1, \ldots, v_m) \in \mathcal{V}_0$} {
        \For{$v_i, i = 1, \ldots, m$} {
            \If{$v_i \in V_{0_i} \setminus V_{goal}$}{
                \textbf{let} $\bigcup_j e_j$ be the set of local transitions starting from $v_i$ in $G_i$\;
                \If{$\bigcup_j e_j \neq \phi$}{
                    clauseList.add([$(\bigvee_{i = 1 \ldots d} \langle v\rangle_i) \Rightarrow (\bigvee_j \langle e\rangle_j)$])\;
                }
            }
        }
    }

    /* STEP 6: Impact of control edge selection (simultaneous progress) */\;
    \For{$v=(v_1, \ldots, v_m) \in \mathcal{V}_0$} {
        \ForAll{edge combination $(e_1, \ldots, e_m)$: $e_i = (v_i, v_i') \in E_i$ when $v_i \in V_{0_i}$ \texttt{or} $e_i = (v_i, v_i)$ when $x_i \in V_{1_i}$}{
        /* $e_i = (v_i, v_i)$ when $x_i \in V_{1_i}$ are simply dummy edges for ease of formulation */\;
            \For{$j = 1, \ldots, d-1$} {
            clauseList.add([$\langle v_1, \ldots, v_m\rangle_j \Rightarrow
            ((\bigwedge_{\{i|v_i \in V_{0_i}\}}\langle e_i\rangle) \Rightarrow$
                $(\langle v_1', \ldots, v_m'\rangle_{j+1}$])
               % $(\langle v_1', \ldots, v_m'\rangle_{j+1} \vee \ldots \vee\langle v_1', \ldots, v_m'\rangle_{d})$])
            }
        }
    }

    /* STEP 7: Impact of environment vertex */\;
    \For{environment vertex $v = (v_1, \ldots, v_m) \in \mathcal{V}_1$} {
        \textbf{let} the set of successors be $\bigcup_i v_i$;
        \For{$j = 1, \ldots, d-1$} {
        %clauseList.add([$\langle v\rangle_j \Rightarrow (\bigwedge_i \langle v_i\rangle_{j+1}) $])
        clauseList.add([$\langle v\rangle_j \Rightarrow (\bigwedge_i (\langle v_i\rangle_{j+1} \vee \ldots \vee, \langle v_i\rangle_{d})) $]);
        }
    }
    /* STEP 8: Invoke the SAT solver: return \texttt{true} when satisfiable */\;
    \textbf{return} \textbf{invokeSATsolver}(clauseList)\;
}
\caption{PositionalDistributedStrategy\_BoundedSAT\_0\label{algo.SAT.DG0}}
\end{algorithm}

%\begin{proof}

%In our argument, we first assert that the satisfiability of clauses generated by
%Algorithm~\ref{algo.SAT.DG0} is a witness. Then we assert that if there exists a witness from initial state
%to goal states, the witness indicates a strategy of the reachability game.

%\end{proof}

In appendix B, we give an alternate algorithm working with different formulation of
distributed games where in each control location, only one local game can move: a run of the game may execute
multiple local moves until it reaches a state where all local games are in an environment position.
We find this alternative formulation closer to the interleaving semantics of distributed systems.

\section{Conversion from Strategies to Concrete Implementations\label{sec.conversion.from.strategy.to.model}}

Once when the distributed game has returned a positive result, and assume that the result is represented as an IM, the remaining
problem is to check whether the synthesized result can be translated to PISEM and thus further to concrete implementation.
If for each existing action or newly generated FT mechanism, the worst case execution
time is known (with available WCET tools, e.g., AbsInt\footnote{\url{http://www.absint.com/}}), then
we can always answer whether the system is implementable
by a full system rescheduling, which can be complicated.
%Based on our system modeling (assumption with globally synchronized clock),
% as in the synthesized IM every action contains a timing precondition based on program counters,
%job scheduling algorithms can be applied to create the new release-time and deadline in the corresponding PISEM:
Nevertheless, based on our system modeling (assumption with a globally synchronized clock), perform modification on the release time  or the deadline on existing actions from the synthesized IM can be translated to a linear constraint system, as in the synthesized
IM each action contains a timing precondition based on program counters.
Here we give a simplified algorithm which performs \emph{local timing modification (LTM)}.
Intuitively, LTM means to perform partitions on either
\begin{enumerate}
    \item the interval $d$ between the deadline of action $\sigma_{\lfloor\frac{a}{b}\rfloor}$ and release time of
            $\sigma_{\lceil\frac{a}{b}\rceil}$, if (a) $\sigma_{\frac{a}{b}}$ exists and (b) $d \neq 0$, or
    \item the execution interval of action $\sigma_{\lfloor\frac{a}{b}\rfloor}$, if $\sigma_{\frac{a}{b}}$ exists.
\end{enumerate}

In the algorithm, we assume that for every action $\sigma_d$, $d\in \mathbf{N}$ where FT mechanisms are not introduced between $\sigma_d$ and $\sigma_{d+1}$ during synthesis, its release-time and deadline should not change; this assumption can be checked later or added explicitly
to the constraint system under solving (but it is not listed here for simplicity reasons).
 Then we solve a constraint system to derive the release time and deadline of all FT actions introduced. Algorithm~\ref{algo.LTM} performs such execution\footnote{Here we list case $2$ only; for case $1$ similar analysis can be applied.}: for simplicity assume at most one FT action exists between two actions $\sigma_i$, $\sigma_{i+1}$; in our implementation this assumption is released:
\begin{list1}
    \item Item (1) performs a interval split between $\sigma_{\lfloor\frac{a}{b}\rfloor}$ and $\sigma_{\frac{a}{b}}$.
    \item Item (3) assigns the deadline of $\sigma_{\lfloor\frac{a}{b}\rfloor}$  to be the original deadline of $\sigma_{\frac{a}{b}}$.
    \item Item (4), (5) ensure that the reserved time interval is greater than the WCET.
    \item Item (6) to (11) introduce constraints from other processes:
    \begin{list1}
        \item Item (6) (7) (8) consider existing actions which do not change the deadline and release time; for these %actions
             fetch the timing information from PISEM.
        \item Item (9) (10) (11) consider newly introduced
        actions or existing actions which change their deadline and release time; for these actions use variables to construct the constraint.
    \end{list1}
    \item Item (12) is a conservative dependency constraint between $\sigma_{\frac{a}{b}}$ and a send $\sigma_{d}$.
\end{list1}

%\vspace{-2mm}
\begin{algorithm}
%\DecMargin{10em}
\DontPrintSemicolon
\KwData{Original PISEM $\mathcal{S} = (\mathcal{A}, \mathcal{N}, \mathcal{T})$, synthesized IM $S = (A, N)$}
\KwResult{ For each $\sigma_{\frac{a}{b}}$ and $\sigma_{\lfloor\frac{a}{b}\rfloor}$, their execution interval
        $[\alpha_{\frac{a}{b}}, \beta_{\frac{a}{b}})$, $[\alpha_{\lfloor\frac{a}{b}\rfloor}, \beta_{\lfloor\frac{a}{b}\rfloor})$}
For convenience, use ($X \;in\;\mathcal{S}$) to represent the retrieved value $X$ from PISEM $\mathcal{S}$.\;
\Begin{
   %\textbf{let} $constraintSystem = getEmptyList()$\;
    \For{ $\sigma_{\frac{a}{b}} [\wedge_{m = 1 \ldots n_A} [pc_{\frac{a}{b},m_{low}},
            pc_{\frac{a}{b},m_{up}})]$ in $\overline{\sigma_i}$ of $A_i$} {
        \textbf{let} $\alpha_{\frac{a}{b}}$, $\beta_{\frac{a}{b}}$, $\alpha_{\lfloor\frac{a}{b}\rfloor}$, $\beta_{\lfloor\frac{a}{b}\rfloor}$
            // Create a new variable for the constraint system \;
        /* Type A constraint: causalities within the process */
        \nl$constraints$.add($\alpha_{\frac{a}{b}} = \beta_{\lfloor\frac{a}{b}\rfloor}$)\;   \label{InRes1}
        \nl$constraints$.add($\alpha_{\lfloor\frac{a}{b}\rfloor} = (\alpha_{\lfloor\frac{a}{b}\rfloor} in\;\mathcal{S})$)\;\label{InRes2}
         \nl$constraints$.add($\beta_{\frac{a}{b}} = (\beta_{\lfloor\frac{a}{b}\rfloor} in\;\mathcal{S})$) \;\label{InRes3}
         \nl$constraints$.add($\beta_{\frac{a}{b}} - \alpha_{\frac{a}{b}} > WCET(\sigma_{\frac{a}{b}})$) \;\label{InRes4}
        \nl$constraints$.add($\beta_{\lfloor\frac{a}{b}\rfloor} - \alpha_{\lfloor\frac{a}{b}\rfloor} >
            WCET(\sigma_{\lfloor\frac{a}{b}\rfloor})$) \;\label{InRe5}
    }
   /* Type B constraint: causalities crossing different processes */\;
   \For{$\sigma_{\frac{a}{b}} [\wedge_{m = 1 \ldots n_A} [pc_{\frac{a}{b},m_{low}},
            pc_{\frac{a}{b},m_{up}})]$ in $\overline{\sigma_i}$ of $A_i$} {
        \For{$\sigma_d [\wedge_{m = 1 \ldots n_A} [pc_{d,m_{low}}, pc_{d,m_{up}})]$ in $\overline{\sigma_j}$ of $A_j$} {
            \eIf{$d\in \mathbf{N}$ and not exists $\sigma_{\frac{x}{y}} \in \overline{\sigma_j}$ where $\lfloor \frac{x}{y}\rfloor = d$}{
                \nl\lIf{$pc_{d,j_{up}} < pc_{\frac{a}{b},j_{low}} $}{$constraints$.add(($\beta_d\;in\;\mathcal{S}) < \alpha_\frac{a}{b}$)}\;\label{InRe6}
                \nl\lIf{$pc_{d,j_{low}} > pc_{\frac{a}{b},j_{up}}$}{$constraints$.add(($\alpha_d\;in\;\mathcal{S}) > \beta_\frac{a}{b}$)}\;\label{InRe7}
                \If{$\sigma_{\frac{a}{b}} := \texttt{send}(pre, ind, n, dest, v, c) \wedge  pc_{d,j_{low}} > pc_{\frac{a}{b},j_{up}}$}{
                    \nl$constraints$.add(($\alpha_d\;in\;\mathcal{S}) > \beta_\frac{a}{b} + WCMTT(n, ind)$)\;\label{InRe8}
                }
            }
            {
                \nl\lIf{$pc_{d,j_{up}} < pc_{\frac{a}{b},j_{low}}$}{$constraints$.add($\beta_d\ < \alpha_\frac{a}{b}$)}\;\label{InRe9}
                \nl\lIf{$pc_{d,j_{low}} > pc_{\frac{a}{b},j_{up}}$}{$constraints$.add($\alpha_d\ > \beta_\frac{a}{b}$)}\;\label{InRe10}
                \If{$\sigma_{\frac{a}{b}} := \texttt{send}(pre, ind, n, dest, v, c) \wedge  pc_{d,j_{low}} > pc_{\frac{a}{b},j_{up}}$}{
                \nl $constraints$.add($\alpha_d > \beta_\frac{a}{b} + WCMTT(n, ind) $) \label{InRe11}
                }
            }
        }
    }
    /* Type C constraint: conservative data dependency constraints */\;
   \For{$\sigma_{\frac{a}{b}} [\wedge_{m = 1 \ldots n_A} [pc_{\frac{a}{b},m_{low}},
            pc_{\frac{a}{b},m_{up}})]$ in $\overline{\sigma_i}$ of $A_i$} {
        \For{$\sigma_d [\wedge_{m = 1 \ldots n_A} (pc_{d,m_{low}}, pc_{d,m_{up}})]$ in $\overline{\sigma_j}$ of $A_j$} {
            \If{$\sigma_{d} := \texttt{send}(pre, ind, n, dest, v, c) \wedge \sigma_{\frac{a}{b}}$ reads variable $c$ $\wedge$ $pc_{d,j_{up}} < pc_{\frac{a}{b},j_{low}}$}{
                \nl$constraints$.add($(\beta_d\;in\;\mathcal{S})+ WCMTT(n, ind) < \alpha_\frac{a}{b}$)\;\label{InRe12}
            }
        }
    }
    \textbf{solve} $constraints$ using (linear) constraint solvers.
}
\caption{LocalTimingModification\label{algo.LTM}}
\end{algorithm}
\section{Implementation and Case Studies \label{sec.case.study}}%: Orchestration of Distributed FT Protocols

For implementation, we have created our prototype software as an Eclipse-plugin,
called \textsc{Gecko}\footnote{\url{http://www6.in.tum.de/~chengch/gecko/}}, which offers an open-platform based on the model-based
approach to facilitate the design, synthesis,
and code generation for fault-tolerant embedded systems.
Currently the engine implements the search-based algorithms, and the SAT-based algorithm is experimented
independently under GAVS\footnote{\url{http://www6.in.tum.de/~chengch/gavs/}}, a tool for visualization and synthesis of games.
%It is released under GPL 3.0, and can be freely modified for both research and educational usages.
% Details maybe later for CAV'11

To evaluate our approach, here we reuse the example in sec.~\ref{sec.motivating.scenarios} and perform automatic
tuning synthesis for the selected FT mechanisms. The models specified in this section, as well as the \textsc{Gecko} Eclipse-plugin
which generates the result, are available in the website.

%(a simple instruction of \textsc{Gecko} for executing examples, screen shots during execution, and the complete case study are in Appendix C).
%For the ease of explanation, from now on, we assume that the action and the network transmission takes no time.

\subsection{Example from Section 2\label{sub.sec.exampleA}}

In this example, the user selects a set of FT mechanism templates with the intention to
implement a \emph{fail-then-resend} operation, which is shown in
Figure~\ref{fig:Motivating.Example.FT1}. The selected patterns introduce two additional messages in the system, and the goal is to orchestrate
multiple synchronization points introduced by the FT mechanisms between $\mathcal{A}$ and $\mathcal{B}$ (the timing in FT mechanisms is unknown). The fault model, similar to sec.~\ref{sec.motivating.scenarios}, assumes that in each period at most one message loss occurs.

\begin{figure}
\centering
 \includegraphics[width=0.8\columnwidth]{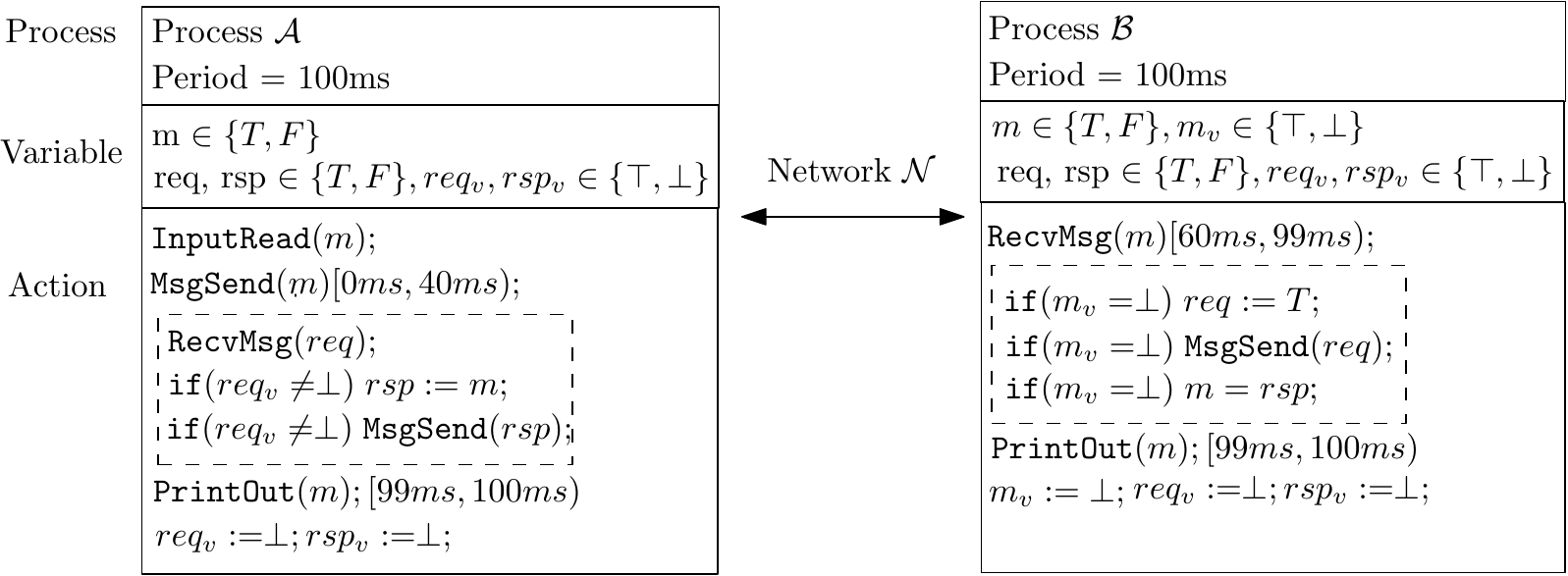}
  \caption{An example where FT primitives are introduced for synthesis.}
 \label{fig:Motivating.Example.FT1}
\end{figure}

Once when \textsc{Gecko} receives the system description (including the fault model) and the reachability specification, it translates the system into
a distributed game. In Figure~\ref{fig:Motivating.Example.FT1.Game}, the set of possible control transitions are listed\footnote{In our implementation,
the PC starts from 0 rather than 1; which is different from the formulation in IM and PISEM.};
the solver generates an appropriate PC-precondition for each action to satisfy the specification. In Figure~\ref{fig:Motivating.Example.FT1.Game},
bold numbers (e.g., $\langle \textbf{0000} \rangle$) indicate the synthesized result. The time line of the
execution (the synthesized result) is explained as follows:
\begin{enumerate}
    \item Process $\mathcal{A}$ reads the input, sends $\verb"MsgSend"(m)$, and waits.
    \item Process $\mathcal{B}$ first waits until it is allowed to execute ($\verb"RecvMsg"(m)$). Then it performs a
    conditional send $\verb"MsgSend"(req)$ and waits.
    \item Process $\mathcal{A}$ performs $\verb"RecvMsg"(req)$, following a conditional send $\verb"MsgSend"(rsp)$.
    \item Process $\mathcal{B}$ performs conditional assignment, which assigns the value of $rsp$ to $m$, if $m_v$ is empty.
\end{enumerate}

\begin{figure}
\centering
 \includegraphics[width=0.8\columnwidth]{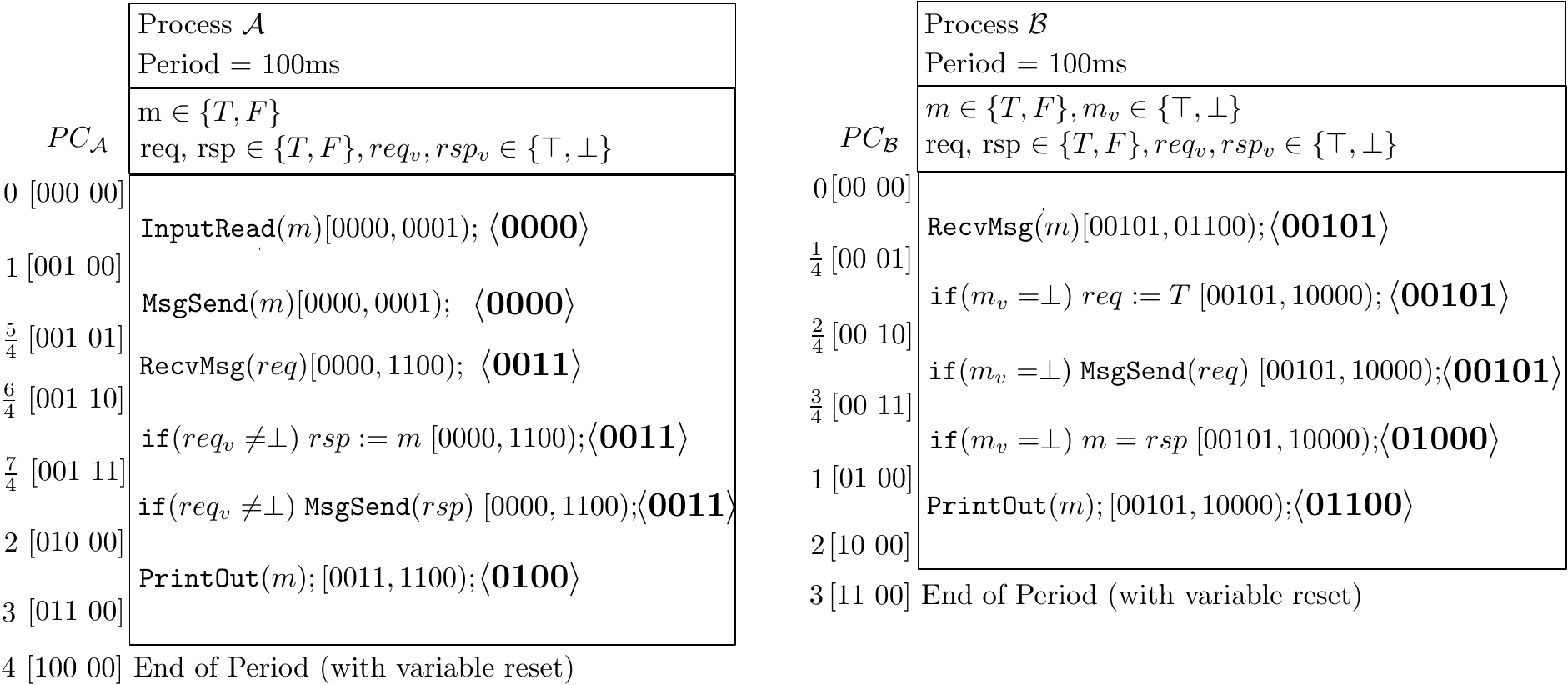}
  \caption{A concept illustration for the control choices in the generated game.}
 \label{fig:Motivating.Example.FT1.Game}
\end{figure}

\begin{figure}
\centering
 \includegraphics[width=0.7\columnwidth]{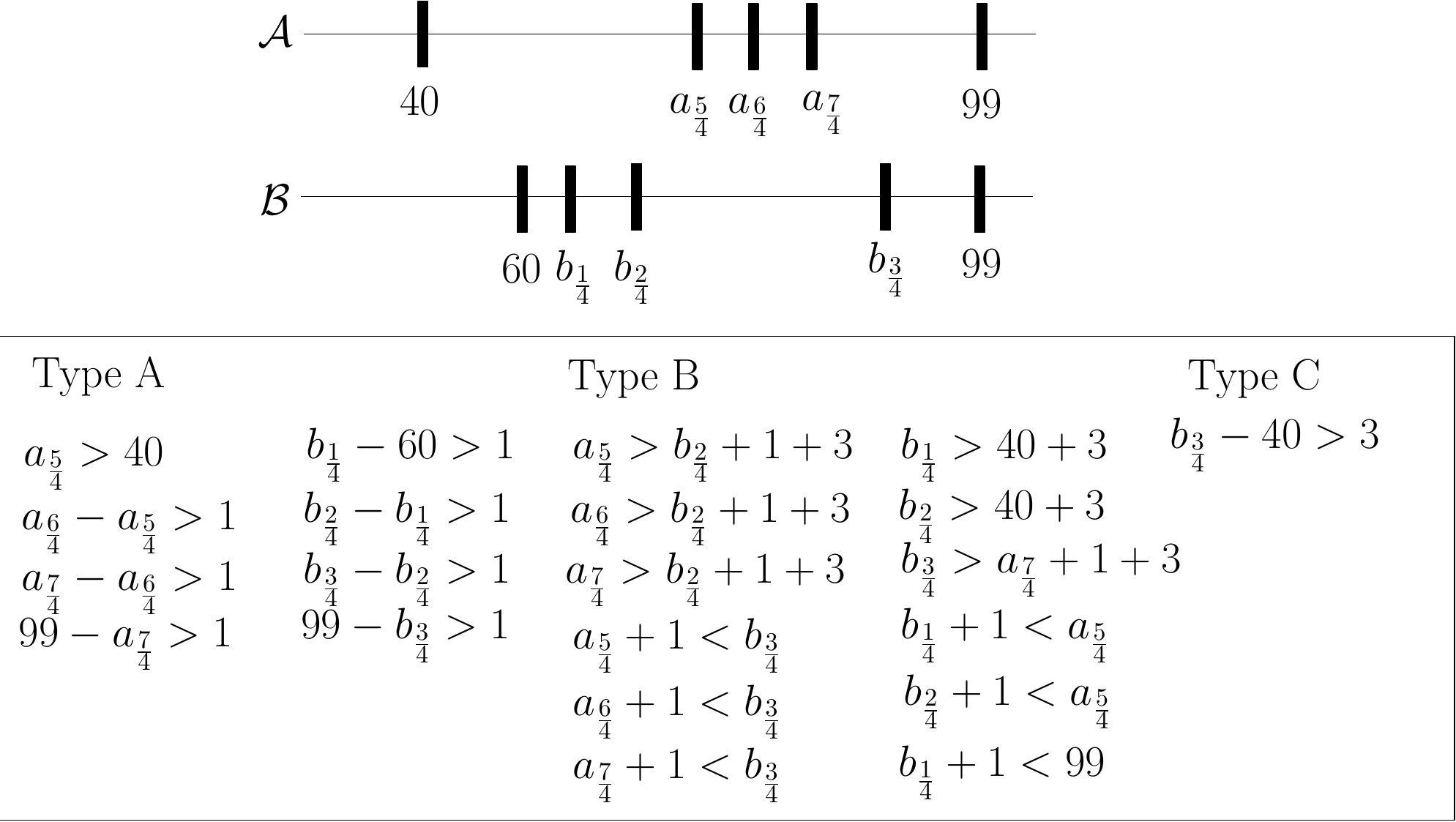}
  \caption{An illustration for applying LTM for the example in sec.~\ref{sub.sec.exampleA}, and the corresponding linear constraints.}
 \label{fig:LSM}
\end{figure}

%In this section, first we continue the process of redirecting the result of synthesis
%back to PISEM for sec.~\ref{sub.sec.exampleA} in C.1, then
%we show an example where when different (inappropriate) FT mechanisms are selected, the game solver still
%generates a satisfying result which is different from the ordering in example A (C.2).
%Lastly, we give instructions how to use \textsc{Gecko} for executing examples specified in this paper.

%\subsection*{C.1: Timing Annotation based on LTM Solving}

We continue the case study by stating assumptions over hardware and timing; these can be specified in \textsc{Gecko} as properties of the model.
\begin{enumerate}
    \item Process $\mathcal{A}$ and $\mathcal{B}$ are running on two Texas Instrument LM3S8962 development
            boards\footnote{\url{http://www.luminarymicro.com/products/LM3S8962.html}} under FreeRTOS\footnote{\url{http://www.freertos.org/ }}
            (a real-time operating system), and messages are communicating over a CAN bus.
    \item For each existing or FT action, its WCET on the hardware is $1$ms.
    \item For all messages communicating using the network, the WCMTT is $3$ms.
\end{enumerate}

We apply the LTM algorithm, such that we can generate timing constraints on dedicated hardware;
these timing constraints will be translated to executable C code (based on FreeRTOS).
Figure~\ref{fig:LSM} is used to assist the explanation of LTM, where variables used in the linear constraint solver are specified as follows:
\begin{list1}
\item $a_{\frac{5}{4}}$: release time for action "$\verb"RecvMsg"(req)$" in process $\mathcal{A}$.
\item $a_{\frac{6}{4}}$: release time for action "$\verb"if"(req_v \neq \perp) \;rsp := m$" (and similarly, the deadline for "$\verb"RecvMsg"(req)$") in process $\mathcal{A}$ .
\item $a_{\frac{7}{4}}$: release time for action "$\verb"if"(req_v \neq \perp) \;\verb"MsgSend"(rsp)$" in process $\mathcal{A}$.
\item $b_{\frac{1}{4}}$: release time for action "$\verb"if"(m_v = \perp) \;req := T$" in process $\mathcal{B}$.
\item $b_{\frac{2}{4}}$: release time for action "$\verb"if"(m_v = \perp) \;\verb"MsgSend"(req) $" in process $\mathcal{B}$.
\item $b_{\frac{3}{4}}$: release time for action "$\verb"if"(m_v = \perp) \;m = rsp$" in process $\mathcal{B}$.
\end{list1}

As in process $\mathcal{A}$, there exists a time interval $[40, 99)$ between two existing actions $\verb"MsgSend"(m)$ and $\verb"PrintOut"(m)$, the LTM
algorithm will prefer to utilize this interval than splitting [0, 40), as using $[40, 99)$ generates the least modification on the scheduling.
The generated linear constraint system is also shown in Figure~\ref{fig:LSM}.
%The only remark concerning the constraint is the first one $a_{\frac{5}{4} \geq 40}$ in Type A.
An satisfying instance for $(a_{\frac{5}{4}}, a_{\frac{6}{4}}, a_{\frac{7}{4}},
b_{\frac{1}{4}}, b_{\frac{2}{4}}, b_{\frac{3}{4}})$ could be $(72, 77, 82, 62, 67, 87)$; instructions concerning the release time and the deadline for the generated fault-tolerant model can be annotated based on this.

\begin{figure}
\centering
 \includegraphics[width=0.75\columnwidth]{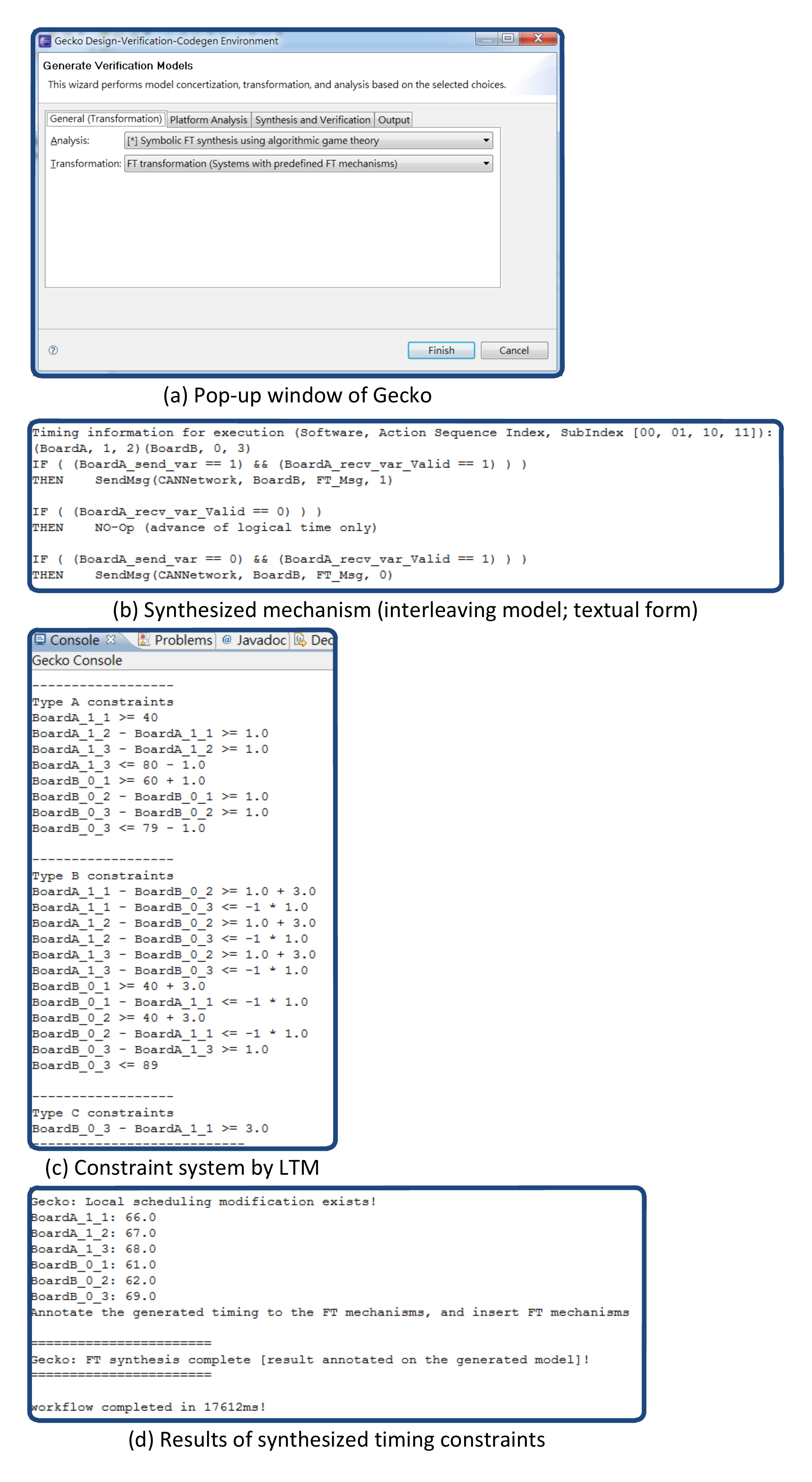}
  \caption{Screenshots of \textsc{Gecko} when executing the example in sec.~\ref{sub.sec.exampleA}.}
 \label{fig:Gecko.Screenshots}
\end{figure}

\subsection{Another Example}
 For the second example, the user selects an inappropriate set of FT mechanisms\footnote{This is originally a design mistake when we specify our FT mechanism patterns; however interesting results are generated.}. Compared to Figure~\ref{fig:Motivating.Example.FT1}, in
 process $\mathcal{A}$ an equality constraint "$\verb"if"(req_v = \perp)$" is used,
 instead of "$\verb"if"(req_v \neq \perp)$". In this way, the combined effect of FT mechanisms in Example B changes
 dramatically from that of Example A:
\begin{list1}
    \item When $\mathcal{B}$ does not receive $m$ from $\mathcal{A}$, it sends a request command.
    \item When $\mathcal{A}$ receives a request message, it does not send the response; this violates the original intention of the designer.
\end{list1}

Surprisingly, \textsc{Gecko} reports a positive result with an interesting sequence! For all FT actions in process $\mathcal{A}$,
they should be executed with the procondition of $PC_\mathcal{B}$ equal to $0000$, meaning that
FT mechanisms in $\mathcal{A}$ are executed before $\verb"RecvMsg"(m)$ in $\mathcal{B}$ starts. In this way,
$\mathcal{A}$ always sends the message $\verb"MsgSend"(rsp)$ containing the value of $m$, and as at most one message loss exists in one period, the specification is satisfied.

\subsection{Discussion}
Concerning the running time of the above two examples, the searching engine (based on forward searching + BDD for intermediate
image storing) is able to report the result in $3$ seconds, while
constraint solving is also relatively fast (within $1$ second). Our engine offers a translation scheme to dump the BDD to
mechanisms in textual form; this process occupies most of the execution time. Note that the NP-completeness result does not
bring huge benefits, as another exponential blow-up caused by the translation from variables to states is also unavoidable:
this is the reason why currently we use a forward search algorithm combining with BDDs in the implementation.

Nevertheless, this does not means that FT synthesis in practice is not possible; our argument is as follows:
\begin{enumerate}
    \item We have indicated that this method is applicable for small examples (similar to the test case in the paper).
    \item To fight with complexity we consider it important to respect the compositional (layered)
        approach used in the design of embedded systems: once when a system have been refined to several subsystems, it is more likely for our approach to be applicable.
\end{enumerate}

\section{Related Work\label{sec.related.work}}

%Concerning related work, here we focus on the
Verification and synthesis of fault tolerance is an active field~\cite{fisman2008verifying,lamport1994specifying,kulkarni2000automating,fmsd:GiraultR09,bernardeschi:2000:fvf,owre:1995:fvf,
arora1998synthesis,dimitrova2009synthesis}.
Among all existing works, we find that the work closest to ours is by Kulkarni et.al.~\cite{kulkarni2000automating}.
Here we summarize the differences in three aspects.
\begin{enumerate}
    \item (Problem) As we are interested in real-time embedded systems, our starting model resembles
        existing formulations used in the real-time community, where time is explicitly stated in the model.
        Their work is more closely to protocol synthesis and the starting model is based on (a composition of) FSMs.
    \item (Approach) As our original intention is to facilitate the pattern selection and tuning process, 
        our approach does not seek for the synthesis of complete FT mechanisms and can be naturally connected to
        games (having a set of predefined moves).
        Contrarily, their results focus on
        synthesizing complete FT mechanisms, for example voting machines or mechanisms for Byzantine generals' problem.
    \item (Algorithm) To apply game-based approach for embedded systems, our algorithms includes the game translation (timing abstraction) and constraint solving (for implementability). In addition, our game formulation enables us to connect and modify
        existing and rich results in algorithmic game solving: for instance, we reuse the idea of
        witness in~\cite{alur:2005:symbolic} for distributed games, and it is likely to establish connections between incomplete methods for distributed games and algorithms for games of imperfect information~\cite{de2006lattice}.
        %Besides, in distributed games, to the best of our knowledge, the research concerning incomplete methods which strives to
        %find (any) strategies is important but missing.
\end{enumerate}
A recent work by Girault et.al.~\cite{fmsd:GiraultR09} follows similar methodologies (i.e., on protocol level FSMs)
to~\cite{kulkarni2000automating} and performs discrete controller synthesis for fault-tolerance; the difference
between our work and theirs follows the argument above.

Lastly, we would like to comment on the application of algorithmic games.
Several important work for game analysis or LTL synthesis can be found
from Bloem and Jobstmann et.al. (the program repair framework~\cite{jobstmann2005program}),
Henzinger and Chatterjee et.al. (Alpaga and the interface synthesis~\cite{berwanger2009alpaga,doyen2008interface}),
or David and Larson et.al. (Uppaal TIGA~\cite{behrmann2007uppaal}).
One important distinction is that due to our system modeling, we naturally start from a problem of solving distributed games
and need to fight with undecidability immediately, while the above works are all based on a non-distributed setting.

\section{Concluding Remarks\label{sec.concluding.remarks}}

\begin{figure}
\centering
 \includegraphics[width=0.8\columnwidth]{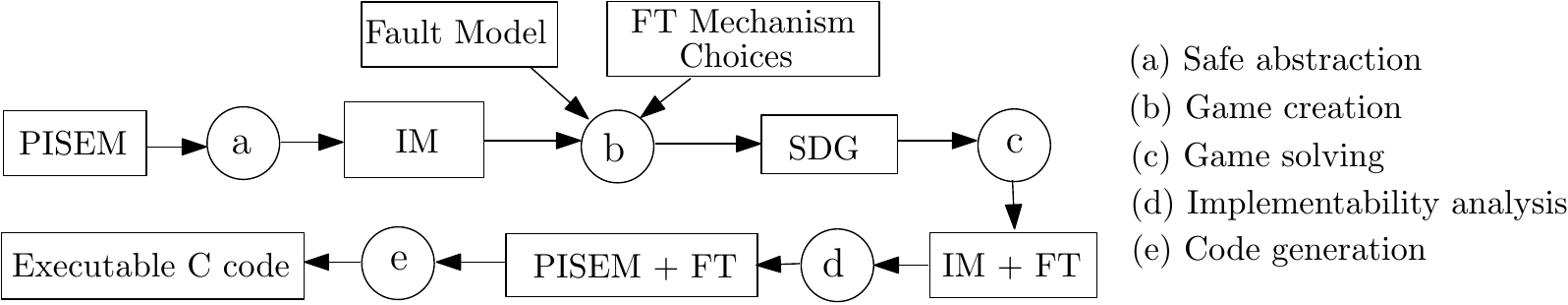}
  \caption{Concept illustration of the overall approach for fault-tolerant synthesis;
  IM+FT means that an IM model is equipped with FT mechanisms.}
 \label{fig:Flow}
\end{figure}

This paper presents a comprehensive approach  (see Figure~\ref{fig:Flow} for concept illustration) for the augmenting of fault-tolerance for real-time distributed systems under a game-theoretic framework.
We use simple yet close-to-reality models (PISEM) as a starting point of FT synthesis,
translate PISEM to distributed games with safe abstractions, perform game solving and later implementability analysis.
The above flow is experimented in a prototype, enabling us to utilize model-based development framework to perform FT synthesis.
These mechanisms may have interesting applications in distributed process control and robotics.
To validate our approach, we plan to increase the maturity of our prototype system and study new algorithms for performance gains.

\subsubsection*{Acknowledgments.} The first author is supported by the DFG Graduiertenkolleg 1480 (PUMA).

\bibliographystyle{abbrv}
\bibliography{refs}


\appendix

\section*{A. The Need of Reestimating the WCMTT in CAN Buses when FT messages are Introduced}

To have an understanding whether newly introduced FT messages can change the existing networking
behavior is both hardware and configuration dependent. In this section, we only describe
the behavior when FT messages are introduced in a Control Area Network (CAN bus),
which is widely used in automotive and automation domains. Here we give configuration settings (conditions)
such that newly introduced messages do \emph{\textbf{not}} influence the existing networking behavior.
For details concerning the timing analysis of CAN, we refer readers to~\cite{tindell1995calculating,davis:2007:controller}.

\begin{prop}
Given an ideal CAN bus with message priority from 1 to $k$, when the three conditions are satisfied:
\begin{enumerate}
    \item No message with priority $k$ is not used in the existing network.
    \item The predefined size of the message for priority $k$ is larger than all messages with priority 1 to $k-1$,
    \item All FT messages are having priority smaller or equal to $k$, and the size is less than the message size stated in (2).
\end{enumerate}
When the WCMTT is derived using the analysis in~\cite{tindell1995calculating}, concerning
the WCMTT of all messages with priority $1$ to $k$, it is indifferent to the newly introduced messages.
\end{prop}

\begin{proof} (Outline) Based on the algorithm in~\cite{tindell1995calculating}, for a message with priority $i \in \{1, \ldots, k\}$,
its timing behavior only changes with two factors:
\begin{itemize}
    \item (a) The blocking time caused by a message with lower priority $j > i$ changes: when a message with lower priority
    changes to a bigger message size, the blocking time increases.
    \item (b) The interference from messages with higher priority $j < i$.
\end{itemize}

We proceed the argument as follows.
 \begin{itemize}
   \item For timing changes due to (b), as FT messages are all with lower priorities (based on condition 3), they do not create or increase interferences with this type.
    \item For timing changes due to (a), we separate two two cases:
    \begin{itemize}
        \item As the size of all FT messages are smaller or equal than the message size specified in (2), then the timing
                behavior for messages with priority 1 to $k-1$ do not change.
        \item Lastly, although the message with priority $k$ can change as it can now be blocked by a lower priority message, such message does not exist based on condition (1).
  \end{itemize}
 \end{itemize}

\end{proof}

By the above information, in our framework we may assume that all messages transmitted in a CAN bus are with lowest priority $k+1$,
and then perform a simple timing analysis at Step A.2.2 before creating the game;
in this way, the problem is Step A.2.2 (\textbf{[Problem 2]}) is safe to neglect.

\section*{B. Algorithm Modification for Interleaving of Local Games}

\begin{algorithm}
\DontPrintSemicolon
\KwData{Distributed game graph $\mathcal{G} = (\mathcal{V}_0 \uplus \mathcal{V}_1, \mathcal{E})$, set of initial states $V_{init}$,
    set of target states $V_{goal}$, the unrolling depth $d$}
\KwResult{Output: whether a distributed positional strategy exists to reach $V_{goal}$ from $v_{init}$}
\Begin{
    \textbf{let} clauseList := \emph{getEmptyList()} /* Store all clauses for SAT solvers */\;
    \textbf{execute} STEP 1 to STEP 5 mentioned in PositionalDistributedStrategy\_BoundedSAT\_0\;

    /* STEP 6: Impact of control edge selection */\;
    \For{ local control transition $e = (x_i, x_i') \in E_i, x_i \in V_{0_i}$} {
        \For{$v=(v_1, \ldots, v_m) \in \mathcal{V}_0 \uplus \mathcal{V}_1$ where $x_i = v_i$} {
            \For{$j = 1, \ldots, d-1$} {
            clauseList.add([$\langle e\rangle \Rightarrow (\langle v_1, \ldots,v_{i-1}, v_{i}, v_{i+1}, \ldots, v_m\rangle_j \Rightarrow$
                $(\langle v_1, \ldots,v_{i-1}, x_i', v_{i+1}, \ldots, v_m\rangle_{j+1} \vee \ldots \vee \langle v_1,
                    \ldots,v_{i-1}, x_i', v_{i+1}, \ldots, v_m\rangle_{d}) )$])
            }
        }
    }

    /* STEP 7: Impact of environment vertex */\;
    \For{environment vertex $v = (v_1, \ldots, v_m) \in \mathcal{V}_1$} {
        \textbf{let} the set of successors be $\bigcup_i v_i$;
        \For{$j = 1, \ldots, d-1$} {
         clauseList.add([$\langle v\rangle_j \Rightarrow (\bigwedge_i (\langle v_i\rangle_{j+1} \vee \ldots \vee, \langle v_i\rangle_{d})) $])
        }
    }
    /* STEP 8: Invoke the SAT solver: return \texttt{true} when satisfiable */\;
    \textbf{return} \textbf{invokeSATsolver}(clauseList)\;
}
\caption{PositionalDistributedStrategy\_ControlInLocalGameInterleaving\_BoundedSAT\_0\label{algo.SAT.DG0.Interleaving}}
\end{algorithm}

\noindent \textbf{(Remark)} Compared to Mohalik and Walukiwitz's formulation, as in this formulation, only one subgame in a control position can
proceed a move, we do not need to create constraints considering all possible combinations in STEP 6, which is required in the algorithm PositionalDistributedStrategy\_BoundedSAT\_0 (sec.~\ref{sec.game.solving}).

\section*{C. Brief Instructions on Executing Examples in \textsc{Gecko}}

%Currently in \textsc{Gecko} the game engine is based on forward searching over FT choices
%(we use JDD\footnote{\url{http://javaddlib.sourceforge.net/jdd/}} to store
%the calculated image); the SAT-based algorithm is implemented independently in GAVS using SAT4J\footnote{\url{http://www.sat4j.org/}}.
%For the implementation of the LTM algorithm, we use %SimplexSolver\footnote{\url{http://commons.apache.org/math/apidocs/org/apache/commons/math/optimization/linear/SimplexSolver.html}}.

Here we illustrate how FT synthesis is done in our prototype tool-chain using the example in sec.~\ref{sub.sec.exampleA}:
first we perform model transformation and generate a new model which equips FT mechanisms.
Then executable code can be generated based on performing code-generation over the specified model (optional).
Once when the \textsc{Gecko} Eclipse add-on is installed (see our website for instructions), proceed with the following steps:
\begin{list1}
    \item The model (\verb"F01_FT_Synthesis_Correct.xmi") for sec.~\ref{sub.sec.exampleA} contains the fault model, the hardware used in the system, and
            pre-inserted FT mechanism blocks, but their timing information is unknown.
    \item Right click on the selected model under synthesis, choose "\verb"Verification"" \verb"->"\\ "\verb"Gecko: Model Transformation and Analysis"". A pop-up window similar to fig.~\ref{fig:Gecko.Screenshots}a is available.
    \item In the \verb"General" tab, choose \verb"Symbolic FT synthesis using"\\ \verb"algorithmic game theory".
    \item In the \verb"Platform Analysis" tab, set up the default actor WCET and network WCMTT to be 1 and 3.
    \item In the \verb"Output" tab, select the newly generated output file.
    \item Press "\verb"Finish"". Results of intermediate steps are shown in the console, including FT mechanisms as interleaving models
            (fig.~\ref{fig:Gecko.Screenshots}b), constraints derived from LTM (fig.~\ref{fig:Gecko.Screenshots}c),
            and results of timing (fig.~\ref{fig:Gecko.Screenshots}d) after executing the constraint solver.
            \begin{list1}
                \item In fig.~\ref{fig:Gecko.Screenshots}b, the mechanism dumped from the engine specifies the action\\ "$\verb"if"(req_v \neq \perp) \;\verb"MsgSend"(rsp)$": note that this action implicitly implies that when $req_v = \perp$, a null-op which only updates the program counter should be executed; this is captured by our synthesis framework.
                \item In fig.~\ref{fig:Gecko.Screenshots}d, the total execution time is roughly 18s because the engine dumps the result back to mechanisms in textual form, which consumes huge amount of time: executing the game and performing constraint solving take only a small portion of the total time.
            \end{list1}
    \item When the model is generated, users can again right click on the newly generated model, and select \verb"Code Generation"
            in the tab \verb"General": the code generator then combines the model description and
            software templates for dedicated hardware and OS to create executable C code.
\end{list1}

\end{document}